\documentclass[aps,pre,showkeys,preprint,12pt,onecolumn,notitlepage,amsmath,amssymb,longbibliography,nofootinbib,floatfix]{revtex4-2}
\usepackage{amssymb, amsmath, amsthm}
\usepackage{mathrsfs}
\usepackage{graphicx}
\graphicspath{{figs/}}

\DeclareMathOperator{\sech}{sech}
\newtheorem{theorem}{Theorem}[section]
\newtheorem{lemma}[theorem]{Lemma}
\newtheorem{proposition}[theorem]{Proposition}

\theoremstyle{definition}
\newtheorem{definition}[theorem]{Definition}

\theoremstyle{remark}
\newtheorem{remark}[theorem]{Remark}

\providecommand{\texorpdfstring}[2]{#1}
\makeatletter
\def\section{%
  \@startsection{section}{1}{\z@}%
    {10pt \@plus 4pt \@minus 2pt}%
    {4pt \@plus 2pt \@minus 2pt}%
    {\normalfont\small\bfseries}%
}
\def\subsection{%
  \@startsection{subsection}{2}{\z@}%
    {8pt \@plus 3pt \@minus 2pt}%
    {3pt \@plus 2pt \@minus 2pt}%
    {\normalfont\small\bfseries}%
}
\def\frontmatter@abstract@produce{%
  \par\vspace{1\baselineskip}%
  \begingroup
    \prep@absbox
    \unvbox\absbox
    \post@absbox
  \endgroup
  \par\vspace{0.5\baselineskip}%
}%
\def\frontmatter@finalspace{}%
\makeatother

\begin{document}

\raggedbottom

\title{Stability analysis of the Lane--Emden equation of index $n=5$ describing polytropic gas spheres}

\author{S.~C.~Mancas}
\email{stefani.mancas.civ@mail.mil}
\affiliation{Laboratory for Physical Sciences, 5520 Research Park Drive, Catonsville, MD 21228, USA}
\author{Ronald Adams}
\email{adamsr25@erau.edu}
\affiliation{Department of Mathematics, Embry-Riddle Aeronautical University, Daytona Beach, FL 32114, USA}


\begin{abstract}
We study the Lane--Emden equation of index $n=5$, the critical polytropic case separating finite-radius from infinite-radius self-gravitating equilibria. The scaling family $\theta_\lambda(\xi)=\lambda^{1/2}\theta(\lambda\xi)$ generates a neutral mode $\phi_0=\tfrac12\theta+\xi\theta'$ about the Schuster profile. Schuster's transformation and a logarithmic variable reduce the equation to the autonomous Hamiltonian system $\ddot\nu=\tfrac14(\nu-\nu^5)$ with a first integral. Besides the Schuster (separatrix) and Srivastava solutions, every orbit crossing $\nu=0$ is given in closed form by a Weierstrass $\wp$-function, and the closed forms agree with direct numerical integration to about $10^{-12}$. Up to scaling and sign, the Schuster solution is the only solution regular at the center, and it separates positive solutions from solutions with infinitely many zeros. In homology-invariant variables the singular solution $\theta=(2\xi)^{-1/2}$ is a center exactly at $n=5$, where its stability changes. For a distance-varying Emden--Fowler system with a shifted equilibrium we prove that at $n=5$ the unique equilibrium is linearly and nonlinearly unstable. Finally, the linearization about the Schuster profile is the Schr\"odinger operator $L=-\Delta_r-5\theta_S^4$ on $L^2((0,\infty),\xi^2d\xi)$. Its Friedrichs extension has essential spectrum $[0,\infty)$, a threshold resonance at $0$ carried by the non-normalizable scaling mode, and exactly one negative eigenvalue, $\lambda_0\approx-1.21$, so the Schuster profile is a saddle point of the associated energy functional.
\end{abstract}
\maketitle
\vspace*{-6pt}

\section{Introduction}
The Lane--Emden equation is the standard model for the equilibrium structure of a spherically symmetric, self-gravitating gas. It follows from hydrostatic equilibrium, Poisson’s equation for gravity and a prescribed equation of state, and in dimensionless form it isolates the nonlinear competition between pressure support and gravitational attraction. Lane derived it in 1870 \cite{Lane} and Emden treated it systematically in 1907 \cite{Emden}; since then it has served stellar structure theory, galactic modeling and the study of nonlinear ordinary differential equations alike.

In its standard form, the Lane--Emden equation of index \(n\) is given by
\begin{equation}
\frac{1}{\xi^{2}} \frac{d}{d\xi} \left( \xi^{2} \frac{d\theta}{d\xi} \right) + \theta^{n} = 0,
\label{LE}
\end{equation}
where \(\xi\) denotes the dimensionless radial coordinate and \(\theta(\xi)\) represents the normalized density or gravitational potential, depending on the physical interpretation. The boundary conditions
\begin{equation}\label{BC}
\theta(0) = 1, \qquad \theta'(0) = 0,
\end{equation}
ensure regularity at the center and correspond to finite central density. For polytropic gas spheres, the physical density is recovered through the relation \(\rho = \rho_{c} \theta^{n}\), where \(\rho_{c}\) denotes the central density.

The physical relevance of the equation is rooted in the polytropic equation of state
\begin{equation}
P = K \rho^{\gamma}, \qquad \gamma = 1 + \frac{1}{n},
\end{equation}
which gives a simple phenomenological description of a wide class of astrophysical systems. Depending on $n$, polytropic models approximate main-sequence stars, white dwarfs, molecular clouds and isothermal gas spheres in various limiting regimes. The case $n=5$, with $\gamma = \tfrac{6}{5}$, is special: it is the critical index separating configurations of finite radius from those extending to infinity, and so it marks the threshold between compact and diffuse equilibria.

Mathematically, the Lane--Emden equation is a difficult nonlinear boundary-value problem. Closed-form solutions are known only for three values of the polytropic index: $n=0$, $n=1$ and $n=5$. The first two give simple algebraic or trigonometric solutions; $n=5$ is richer, with whole families of exact solutions generated by nontrivial transformations. In particular, the classical solution
\begin{equation}
\theta(\xi) = \left(1 + \frac{\xi^{2}}{3}\right)^{-1/2}
\end{equation}
exhibits algebraic decay and infinite spatial extent, while the total mass of the corresponding configuration is finite.

The special status of the $n=5$ case was recognized already in the late nineteenth century. Schuster’s pioneering work \cite{Sch} introduced transformations revealing hidden symmetries of the equation, later developed by Fowler \cite{Fow14,Fow30,Fow31} and incorporated into Chandrasekhar’s comprehensive theory of stellar structure \cite{Chan}. Beyond the Schuster profile, Srivastava~\cite{Srivastava} and Sharma~\cite{Sharma1977} obtained additional exact branches; Goenner and Havas~\cite{GoennerHavas2000} expressed oscillatory solutions through Weierstrass $\wp$-functions; and Mach~\cite{Mach2012} classified \emph{all} real solutions in closed form using Jacobian and Weierstrass elliptic functions (see also \cite{Lawden,Olver}). Integrable reductions of related Emden--Fowler equations, including the $n=5$ Lane--Emden case, have been studied systematically by Mancas and Rosu~\cite{MancasRosu2018}, while Weierstrass-based integration techniques for nonlinear wave equations appear in~\cite{MancasBBM2013}. The solution space at $n=5$ is thus far more intricate than for the lower indices.

Stability of the $n=5$ solutions, and of the Schuster profile on the separatrix of the elliptic-function family in particular, is a separate matter. The question is whether a configuration persists under small perturbations or collapses or disperses, which in astrophysical terms separates realizable gas spheres from equilibria that exist only on paper. Mathematically it comes down to the spectrum of the linearized operator and the global phase-space picture.

Here we analyze the stability of the index-$n=5$ equation, tying together its Weierstrass solution structure, its global phase portrait and its dynamical stability. After reducing the equation to autonomous first-integral form, we express generic bounded orbits through $\wp$-functions (Section~\ref{sec:weierstrass}) and contrast them with the algebraic Schuster and Srivastava limits. With Lyapunov methods and phase-space analysis we decide which configurations are stable or unstable and characterize perturbations about the Schuster profile~\cite{Polytropes,Davis}.

The equation has also been revisited recently through symmetry methods, elliptic integrability and dynamical systems theory. B\"ohmer and Harko~\cite{BohmerHarko2010} investigated generalized polytropic models; Adams, Mancas, and Rosu~\cite{AMR2018} studied the stability of orbital modes for generalized Lane--Emden equations; and Mach~\cite{Mach2012} completed the closed-form classification of $n=5$ solutions. The present paper combines such exact representations with a stability analysis.

The paper is organized as follows. Section~\ref{sec:analytic} reduces the equation to an autonomous Hamiltonian system and collects the Schuster, Srivastava and Weierstrass solutions, together with the homology-invariant phase plane and Mach's classification. Section~\ref{sec:stability} proves linear and nonlinear instability of the equilibrium of the distance-varying system. Section~\ref{sec:variational} gives the variational structure, and Section~\ref{sec:linear_stability} studies the spectrum of the linearized operator about the Schuster profile. Sections~\ref{sec:future} and~\ref{sec:conclusion} discuss open directions and summarize the results.

\subsection{Scaling Invariance at \texorpdfstring{$n=5$}{n=5}}
\begin{proposition}[Scaling law and critical index]\label{prop:scaling}
If $\theta$ satisfies \eqref{LE} for a given $n$, then for every $\lambda>0$ the rescaled function
\begin{equation}\label{eq:scaling_family}
\theta_\lambda(\xi)=\lambda^{\frac{2}{n-1}}\theta(\lambda\xi)
\end{equation}
also satisfies \eqref{LE} with the same index $n$. In particular, when $n=5$,
\begin{equation}\label{eq:scaling_symmetry_intro}
\theta_\lambda(\xi)=\lambda^{1/2}\theta(\lambda\xi).
\end{equation}
For a solution normalized by \eqref{BC}, the scaled function satisfies $\theta_\lambda(0)=1$ if and only if
\begin{equation}\label{eq:scaling_norm}
\lambda^{\frac{2}{n-1}}=1,
\end{equation}
which for $\lambda>0$ forces $\lambda=1$ for every polytropic index $n$. The distinguished role of $n=5$ is not normalization invariance, but the critical similarity exponent $2/(n-1)=1/2$ and the resulting infinitesimal symmetry
\begin{equation}\label{eq:phi0_intro}
\phi_0(\xi)=\left.\frac{d}{d\lambda}\theta_\lambda(\xi)\right|_{\lambda=1}
=\frac{1}{2}\theta(\xi)+\xi\theta'(\xi),
\end{equation}
which is a formal null-mode of the linearized operator about the Schuster profile (Section~\ref{sec:linear_stability}).
\end{proposition}

\begin{proof}
Substituting \eqref{eq:scaling_family} into \eqref{LE} and writing $\eta=\lambda\xi$ gives
\[
\frac{1}{\xi^2}\frac{d}{d\xi}\!\left(\xi^2\frac{d\theta_\lambda}{d\xi}\right)
+\lambda^{\frac{2n}{n-1}}\theta(\lambda\xi)^n
=\lambda^{\frac{2n}{n-1}}\left[
\frac{1}{\eta^2}\frac{d}{d\eta}\!\left(\eta^2\frac{d\theta}{d\eta}\right)+\theta(\eta)^n
\right]_{\eta=\lambda\xi}=0.
\]
At $\xi=0$, $\theta_\lambda(0)=\lambda^{2/(n-1)}\theta(0)=\lambda^{2/(n-1)}$, so \eqref{BC} is preserved only when \eqref{eq:scaling_norm} holds; for $\lambda>0$ this gives $\lambda=1$. Differentiating \eqref{eq:scaling_family} at $\lambda=1$ yields \eqref{eq:phi0_intro}. At $n=5$ the exponent $1/2$ is the critical polytropic similarity index marking the transition between finite- and infinite-radius models; the tangent $\phi_0$ is nontrivial for the Schuster solution and generates the formal scaling null-mode analyzed in Section~\ref{sec:linear_stability} (see Figure~\ref{fig:scaling_family}).
\end{proof}

\begin{figure}[htbp]
\centering
\includegraphics[width=0.72\linewidth]{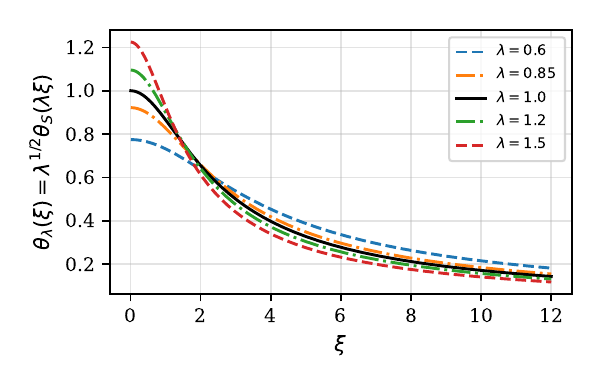}\\[0.6em]
\caption{Members of the scaling family \eqref{eq:scaling_symmetry_intro} at $n=5$, generated from the Schuster profile. Each curve is a solution of \eqref{LE}; only $\lambda=1$ satisfies $\theta(0)=1$, while varying $\lambda$ illustrates the continuous symmetry whose derivative at $\lambda=1$ is $\phi_0$.}
\label{fig:scaling_family}
\end{figure}

\section{Analytic Solutions of the Lane--Emden Equation of Index \texorpdfstring{$n=5$}{n=5}}
\label{sec:analytic}
The analytical study of the Lane--Emden equation traces back to Lane in 1870, who investigated the internal structure of the Sun under the assumption that a gaseous mass could maintain equilibrium through internal heat generation \cite{Lane}. Lane’s formulation combined Newtonian gravity with hydrostatic equilibrium and a polytropic equation of state, laying the foundations for modern stellar structure theory.

Consider a self-gravitating, spherically symmetric gas sphere in hydrostatic equilibrium, governed by the coupled system
\begin{align}
\frac{dP}{dr} &= -\frac{G M(r)\rho(r)}{r^2}, \label{hydrostatic}\\
\frac{dM(r)}{dr} &= 4\pi r^2 \rho(r), \label{mass}\\
P &= K \rho^{\gamma}, \qquad \gamma = 1 + \frac{1}{n}, \label{polytrope}
\end{align}
where $P(r)$, $\rho(r)$, and $M(r)$ denote the pressure, density, and enclosed mass at radius $r$, respectively, and $K$ is a constant determined by the thermodynamic properties of the gas.

Introducing the dimensionless variables
\begin{equation}
\xi = \frac{r}{\alpha}, \qquad
\theta(\xi) = \left(\frac{\rho(r)}{\rho_c}\right)^{1/n},
\end{equation}
where $\rho_c$ is the central density and
\begin{equation}
\alpha = \left[\frac{(n+1)K\rho_c^{(1-n)/n}}{4\pi G}\right]^{1/2},
\end{equation}
the system \eqref{hydrostatic}–\eqref{polytrope} reduces to the Lane--Emden equation \eqref{LE}
subject to the regularity conditions \eqref{BC}.

In 1907, Emden \cite{Emden} systematically studied generalized equations of the form
\begin{equation}\label{eq1}
z_{\chi\chi}=\pm \chi^{1-n} z^n,
\end{equation}
while Fowler \cite{Fow30} later introduced the broader class of equations
\begin{equation}\label{eq2}
\frac{d}{dx}\left(x^2\frac{dy}{dx}\right)+x^{\lambda}y^n=0,
\end{equation}
now known as the Emden--Fowler equations. Equation \eqref{eq1} may be obtained from \eqref{eq2} through the transformations
\begin{equation}
y(x)=z(\chi), \qquad \chi=\frac{1}{x}, \qquad \lambda=n-3.
\end{equation}

The special case $\lambda=2$ corresponds precisely to the Lane--Emden equation of index $n=5$, which may be written in self-adjoint form as
\begin{equation}\label{eq3}
\frac{1}{x^2}\frac{d}{dx}\left(x^2\frac{dy}{dx}\right)+y^5=0.
\end{equation}
This equation is equivalent to the form derived earlier by Schuster \cite{Sch},
\begin{equation}\label{eq4}
z_{\chi\chi}+\chi^{-4}z^5=0.
\end{equation}

Following Schuster’s approach, once a particular solution $z_p(\chi)$ of \eqref{eq4} is known, a more general solution may be constructed by setting
\begin{equation}
z(\chi)=z_p(\chi)\nu(\chi).
\end{equation}
Substituting this ansatz into \eqref{eq4} yields
\begin{equation}\label{eq5}
z_p\nu_{\chi\chi}+2z_p'\nu_\chi+z_p''(\nu-\nu^5)=0.
\end{equation}

According to Kamke \cite{Kamke}, equation \eqref{eq4} admits the particular solution
\begin{equation}
z_p(\chi)=\frac{1}{\sqrt{2}}\sqrt{\chi},
\end{equation}
as may be verified by direct substitution. Introducing this into \eqref{eq5} and simplifying yields the nonlinear Cauchy--Euler equation
\begin{equation}\label{eq6}
\chi^2\nu_{\chi\chi}+\chi\nu_\chi-\frac{1}{4}(\nu-\nu^5)=0.
\end{equation}

The damping term can be eliminated by the logarithmic transformation
\begin{equation}
t=\log \chi,
\end{equation}
which converts \eqref{eq6} into the autonomous second-order equation
\begin{equation}\label{eq7}
\ddot{\nu}-\frac{1}{4}(\nu-\nu^5)=0.
\end{equation}
\begin{lemma}[First integral of the reduced dynamics]\label{lem:first_integral}
Every solution of \eqref{eq7} satisfies
\begin{equation}\label{eq8}
\frac{1}{2}\dot{\nu}^2=\frac{1}{8}\nu^2-\frac{1}{24}\nu^6+\nu_0,
\end{equation}
where $\nu_0\in\mathbb{R}$ is constant along trajectories.
\end{lemma}

\begin{proof}
Multiply \eqref{eq7} by $\dot{\nu}$. Using $\ddot{\nu}=\tfrac14(\nu-\nu^5)$ from \eqref{eq7},
\[
\frac{d}{dt}\left(\frac{1}{2}\dot{\nu}^2\right)
=\dot{\nu}\,\ddot{\nu}
=\dot{\nu}\,\frac{1}{4}\left(\nu-\nu^5\right)
=\frac{d}{dt}\left(\frac{\nu^2}{8}-\frac{\nu^6}{24}\right).
\]
Thus $\tfrac12\dot{\nu}^2-\tfrac18\nu^2+\tfrac{1}{24}\nu^6$ is conserved; the equivalent energy integral in Schuster's original variables is established in~\cite{Mach2012}, while Chandrasekhar~\cite[Ch.~IV, p.~96]{Chan} provides the analytic $n=5$ solution $\theta_S$ and the polytrope mass formula.
\end{proof}

The parameter $\nu_0$ labels distinct energy levels of the autonomous system \eqref{eq7} (Figure~\ref{fig:energy_levels}). The level $\nu_0=0$ is the separatrix and carries the Schuster solution; the Srivastava level and the generic levels are treated afterwards.

\begin{figure}[htbp]
\centering
\includegraphics[width=0.72\linewidth]{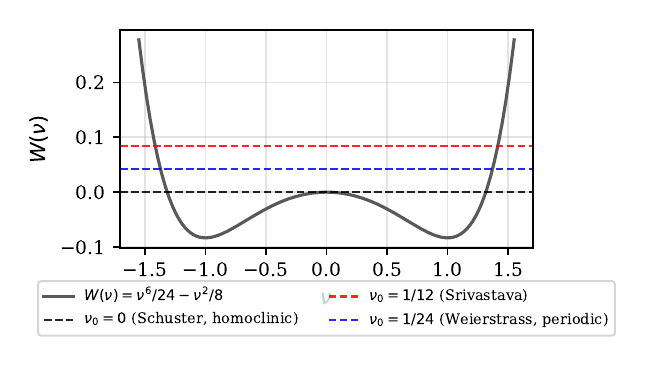}\\[0.6em]
\caption{Effective potential $W(\nu)=\nu^6/24-\nu^2/8$ (gray) and horizontal energy levels $\nu_0$ for the first integral \eqref{eq8}; the central maximum at $\nu=0$ is the hyperbolic saddle of Theorem~\ref{thm:saddle_nu}. Highlighted: the separatrix $\nu_0=0$ (black), corresponding to the Schuster solution, whose trajectory is \emph{homoclinic} -- a single orbit that leaves and returns to the same equilibrium as $t\to\pm\infty$, rather than repeating; the Srivastava level $\nu_0=1/12$ (red), a smooth periodic orbit in elementary form; and a generic level $\nu_0=1/24$ (blue), a periodic orbit expressed through Weierstrass elliptic functions (Section~\ref{sec:weierstrass}). Levels $-1/12<\nu_0<0$ give periodic orbits inside one of the two wells, levels $\nu_0>0$ periodic orbits that cross $\nu=0$. These colors are used consistently in Figures~\ref{fig:integrable_nu}, \ref{fig:weierstrass_family}, and~\ref{fig:phase_portrait}.}
\label{fig:energy_levels}
\end{figure}

\subsection{The Classical Schuster Solution (\texorpdfstring{$\nu_0=0$}{nu0=0})}
\label{sec:schuster}
Setting $\nu_0=0$ in the first integral \eqref{eq8} yields
\begin{equation}\label{eq:schuster_energy}
\frac{1}{2}\dot{\nu}^2
=\frac{1}{8}\nu^2-\frac{1}{24}\nu^6
=\frac{1}{24}\nu^2(3-\nu^4).
\end{equation}
Real trajectories therefore satisfy $|\nu|\le 3^{1/4}$. Equation
\eqref{eq:schuster_energy} represents the separatrix energy level of the
associated autonomous system \eqref{eq7}.

\medskip
\noindent
\textbf{Separation of variables.}
From \eqref{eq:schuster_energy} we obtain
\begin{equation}\label{eq:schuster_sep}
\dot{\nu}
=\pm\frac{1}{\sqrt{12}}\nu\sqrt{3-\nu^4},
\end{equation}
and hence
\begin{equation}\label{INT}
\int
\frac{2\,d\nu}{\nu\sqrt{1-\nu^4/3}}
=
t-t_0.
\end{equation}
Introduce
\begin{equation}\label{eq:sub_u}
u=\frac{\nu^2}{\sqrt{3}}, \qquad u>0.
\end{equation}
Since $\nu^2=\sqrt{3}\,u$ and $d\nu=\frac{\sqrt{3}}{2\nu}\,du$, the integrand simplifies to
\[
\frac{2\,d\nu}{\nu\sqrt{1-\nu^4/3}}
=
\frac{1}{u}\frac{du}{\sqrt{1-u^2}}.
\]
Thus \eqref{INT} becomes
\begin{equation}\label{eq:int_u}
\int \frac{du}{u\sqrt{1-u^2}}=t-t_0.
\end{equation}

Now set $u=\sin\vartheta$, $0<\vartheta<\pi/2$. Then
\[
\frac{du}{u\sqrt{1-u^2}}
=
\frac{\cos\vartheta\,d\vartheta}{\sin\vartheta\,\cos\vartheta}
=
\csc\vartheta\,d\vartheta,
\]
and therefore
\begin{equation}
\int \csc\vartheta\,d\vartheta=t-t_0.
\end{equation}
Using the classical identity
\[
\int \csc\vartheta\,d\vartheta
=
\ln\!\left|\tan\frac{\vartheta}{2}\right|+C,
\]
we obtain
\begin{equation}\label{eq:theta_sol}
\ln\!\left(\tan\frac{\vartheta}{2}\right)=t-t_0.
\end{equation}

Exponentiation gives
\[
\tan\frac{\vartheta}{2}=e^{\,t-t_0},
\]
so that
\begin{equation}\label{eq:theta_explicit}
\vartheta(t)=2\arctan\!\bigl(e^{\,t-t_0}\bigr).
\end{equation}

\medskip
\noindent
\textbf{Explicit form of $\nu(t)$.}
Since $\nu^2=\sqrt{3}\sin\vartheta$, we obtain
\begin{equation}
\nu(t)
=
3^{1/4}
\sqrt{\sin\!\bigl(2\arctan(e^{\,t-t_0})\bigr)}.
\end{equation}
Using the identity
\[
\sin(2\arctan y)=\frac{2y}{1+y^2},
\]
this simplifies to
\begin{equation}\label{eq:nu_simplified}
\nu(t)
=
3^{1/4}
\sqrt{\frac{2e^{\,t-t_0}}{1+e^{\,2(t-t_0)}}}
=
3^{1/4}\sqrt{\sech(t-t_0)}\end{equation}

\medskip
\noindent
\textbf{Recovery of the Lane--Emden profile.}
Recalling that $t=\log\chi$ and $z(\chi)=z_p(\chi)\nu(\chi)$ with
\[
z_p(\chi)=\frac{1}{\sqrt{2}}\chi^{1/2},
\]
substitution of \eqref{eq:nu_simplified} yields
\[
z(\chi)
=
\frac{1}{\sqrt{2}}\chi^{1/2}
\,3^{1/4}\sqrt{\sech(\log\chi-\log\chi_0)}.
\]
Setting $\chi_0=1/\sqrt{3}$ and transforming back to the Lane--Emden variable $\xi$ gives the classical closed-form solution
\begin{equation}
\theta(\xi)=\left(1+\frac{\xi^2}{3}\right)^{-1/2}.
\end{equation}

This is the separatrix configuration of the $n=5$ polytrope. It decays algebraically as $\theta(\xi)\sim \sqrt{3}\,\xi^{-1}$ for large $\xi$, has infinite spatial extent, and sits at the threshold between finite-radius and infinite-radius equilibrium configurations (Figures~\ref{fig:schuster_profile} and~\ref{fig:numeric_comparison}).

\begin{figure}[htbp]
\centering
\includegraphics[width=0.72\linewidth]{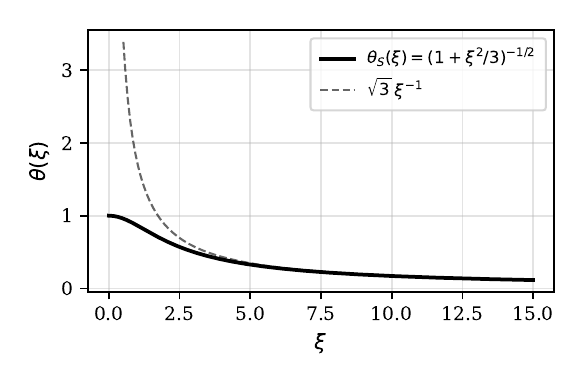}\\[0.6em]
\caption{The Schuster profile $\theta_S(\xi)=(1+\xi^2/3)^{-1/2}$ (solid) and the asymptotic reference $\sqrt{3}\,\xi^{-1}$ (dashed).}
\label{fig:schuster_profile}
\end{figure}

\begin{figure}[htbp]
\centering
\includegraphics[width=0.72\linewidth]{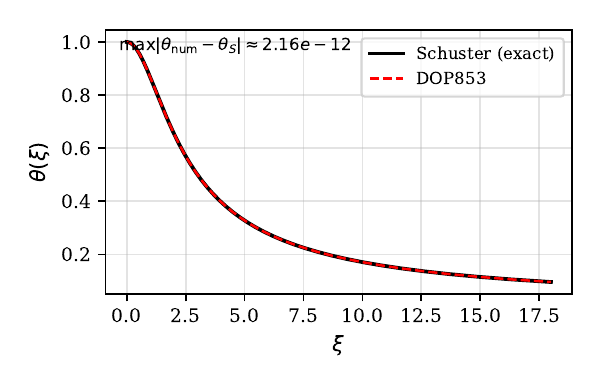}\\[0.6em]
\caption{Direct DOP853 integration of \eqref{eq:lane_emden_n5} (dashed; relative tolerance $10^{-12}$, started from the regular series at $\xi=10^{-4}$) compared with the closed-form Schuster solution (solid); maximum deviation $2.2\times10^{-12}$ on $0\le\xi\le18$.}
\label{fig:numeric_comparison}
\end{figure}

\subsection{The Srivastava Solution (\texorpdfstring{$\nu_0=\tfrac{1}{12}$}{nu0=1/12})}
\label{sec:srivastava}
Srivastava~\cite{Srivastava} identified a second integrable case of \eqref{eq8}, corresponding to the energy level $\nu_0=\tfrac{1}{12}$. Substituting into \eqref{eq8} gives
\begin{equation}\label{eq:sriv_start}
\frac12\dot{\nu}^2=\frac18\nu^2-\frac{1}{24}\nu^6+\frac{1}{12}.
\end{equation}
Multiplying \eqref{eq:sriv_start} by $24$ yields
\begin{equation}\label{eq:sriv_poly24}
12\dot{\nu}^2=3\nu^2-\nu^6+2.
\end{equation}
The sextic polynomial on the right-hand side factorizes conveniently:
\begin{equation}\label{eq:sriv_factor}
3\nu^2-\nu^6+2
=
(2-\nu^2)(1+\nu^2)^2,
\end{equation}
which can be verified by direct expansion:
\[
(2-\nu^2)(1+\nu^2)^2=(2-\nu^2)(1+2\nu^2+\nu^4)=2+3\nu^2-\nu^6.
\]
Hence \eqref{eq:sriv_poly24} becomes
\begin{equation}\label{eq:sriv_sep}
12\dot{\nu}^2=(2-\nu^2)(1+\nu^2)^2,
\qquad\text{so}\qquad
\dot{\nu}=\pm\frac{1}{\sqrt{12}}(1+\nu^2)\sqrt{2-\nu^2}.
\end{equation}
Real trajectories require $2-\nu^2\ge0$, so $|\nu|\le\sqrt{2}$. Separating variables in \eqref{eq:sriv_sep} gives
\begin{equation}\label{eq:sriv_int}
\int \frac{d\nu}{(1+\nu^2)\sqrt{2-\nu^2}}
=\pm\frac{t-t_0}{\sqrt{12}}.
\end{equation}

To evaluate the integral explicitly, set
\begin{equation}\label{eq:sriv_sub}
\nu=\sqrt{2}\,\sin\vartheta,\qquad \vartheta\in\Bigl(-\frac{\pi}{2},\frac{\pi}{2}\Bigr),
\end{equation}
so that $d\nu=\sqrt{2}\cos\vartheta\,d\vartheta$ and $\sqrt{2-\nu^2}=\sqrt{2}\cos\vartheta$. Substituting into \eqref{eq:sriv_int} yields
\begin{equation}\label{eq:sriv_theta_int}
\int \frac{\sqrt{2}\cos\vartheta\,d\vartheta}{\bigl(1+2\sin^2\vartheta\bigr)\sqrt{2}\cos\vartheta}
=
\int\frac{d\vartheta}{1+2\sin^2\vartheta}
=\pm\frac{t-t_0}{\sqrt{12}}.
\end{equation}
Using $\tan\vartheta = u$ so that $\sin^2\vartheta=\frac{u^2}{1+u^2}$ and $d\vartheta=\frac{du}{1+u^2}$, we obtain
\begin{equation}\label{eq:sriv_u_int}
\int\frac{d\vartheta}{1+2\sin^2\vartheta}
=
\int \frac{1}{1+\frac{2u^2}{1+u^2}}\frac{du}{1+u^2}
=
\int \frac{du}{1+3u^2}
=
\frac{1}{\sqrt{3}}\arctan(\sqrt{3}\,u).
\end{equation}
Therefore \eqref{eq:sriv_theta_int} becomes
\begin{equation}\label{eq:sriv_arctan}
\frac{1}{\sqrt{3}}\arctan\!\bigl(\sqrt{3}\tan\vartheta\bigr)
=\pm\frac{t-t_0}{\sqrt{12}}.
\end{equation}
Multiplying by $\sqrt{3}$ and noting $\sqrt{3}/\sqrt{12}=1/2$, we obtain
\begin{equation}\label{eq:sriv_simple}
\arctan\!\bigl(\sqrt{3}\tan\vartheta\bigr)=\pm\frac{t-t_0}{2}.
\end{equation}
Taking the tangent of both sides gives
\begin{equation}\label{eq:sriv_tan}
\sqrt{3}\tan\vartheta=\tan\!\left(\pm\frac{t-t_0}{2}\right),
\qquad\text{hence}\qquad
\tan\vartheta=\frac{1}{\sqrt{3}}\tan\!\left(\pm\frac{t-t_0}{2}\right).
\end{equation}
Finally, $\nu=\sqrt{2}\sin\vartheta$ with $\sin\vartheta=\tan\vartheta/\sqrt{1+\tan^2\vartheta}$ gives, for $|t-t_0|<\pi$,
$\nu=\sqrt{2}\tan\bigl(\tfrac{t-t_0}{2}\bigr)\big/\sqrt{3+\tan^2\bigl(\tfrac{t-t_0}{2}\bigr)}$. This expression has jumps from $\sqrt2$ to $-\sqrt2$ at $t-t_0=(2k+1)\pi$, which are an artifact of the branch of $\arctan$ in \eqref{eq:sriv_tan}: continuing $\vartheta$ continuously through $\pm\pi/2$ (multiply numerator and denominator by $\cos\frac{t-t_0}{2}$) yields the smooth solution
\begin{equation}\label{eq:sriv_final}
\nu(t)=\frac{\sqrt{2}\,\sin\!\left(\frac{t-t_0}{2}\right)}
{\sqrt{1+2\cos^2\!\left(\frac{t-t_0}{2}\right)}},
\end{equation}
valid for all $t\in\mathbb{R}$. It is $4\pi$-periodic, satisfies $|\nu(t)|\le\sqrt{2}$, and attains $\pm\sqrt2$ (the turning points of \eqref{eq:sriv_sep}) at $t-t_0=(2k+1)\pi$. In the Lane--Emden variables, $\theta(\xi)=\nu/\sqrt{2\xi}$ with $t=-\log\xi$, this is Srivastava's oscillating solution $\theta(\xi)=\sin(\ln\sqrt{\xi})\big/\bigl[\sqrt{\xi}\,\bigl(3-2\sin^2(\ln\sqrt{\xi})\bigr)^{1/2}\bigr]$ (up to sign and the choice of $t_0$), with infinitely many zeros. A direct DOP853 integration of \eqref{eq7} from $\nu(t_0)=0$ on the level $\nu_0=\tfrac1{12}$ reproduces \eqref{eq:sriv_final} to $1.2\times10^{-12}$ on $|t-t_0|\le10$ (Figure~\ref{fig:integrable_nu}). The Srivastava orbit thus oscillates indefinitely, in contrast to the Schuster solution, which decays to $0$ as $t\to\pm\infty$ on the separatrix level $\nu_0=0$.%

\begin{figure}[htbp]
\centering
\includegraphics[width=0.72\linewidth]{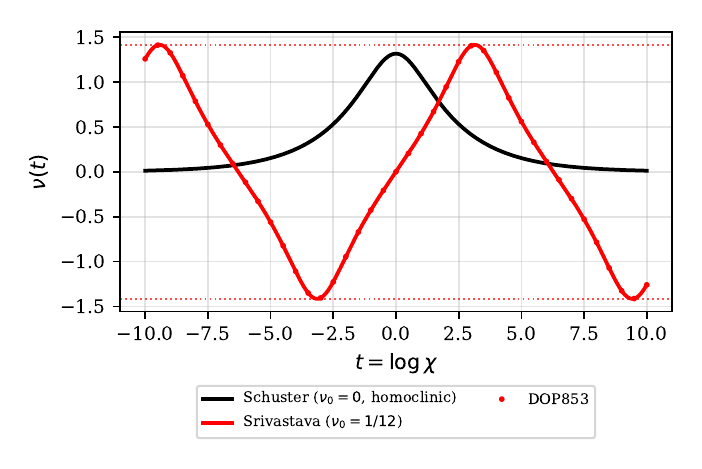}\\[0.6em]
\caption{Explicit integrable solutions $\nu(t)$ for $\nu_0=0$ (Schuster, black, homoclinic; see Figure~\ref{fig:energy_levels}) and $\nu_0=1/12$ (Srivastava, red, elementary) in the logarithmic variable $t=\log\chi$, color-matched to Figures~\ref{fig:energy_levels}, \ref{fig:weierstrass_family}, and~\ref{fig:phase_portrait}. The Srivastava curve \eqref{eq:sriv_final} is smooth and $4\pi$-periodic, with turning points $\pm\sqrt{2}$ (horizontal dotted lines); dots: direct DOP853 integration of \eqref{eq7} (maximum deviation $1.2\times10^{-12}$).}
\label{fig:integrable_nu}
\end{figure}
\subsection{General Properties of the Integrable Cases}
The two integrable solutions have the following features in common or in contrast:
\begin{enumerate}
\item [i)] Boundedness: Both solutions remain bounded for all $t\in\mathbb{R}$.
\item [ii)] Asymptotic behavior: The Schuster solution decays to $0$ as $t\to\pm\infty$ (algebraic decay of $\theta_S$ in $\xi$), while the Srivastava solution oscillates periodically between $\pm\sqrt2$.
\item [iii)] Phase-space structure: Each solution lies on a distinct invariant energy manifold of \eqref{eq7}.
\item [iv)] Stability relevance: These exact solutions provide natural reference states for subsequent linear and nonlinear stability analysis.
\end{enumerate}

\subsection{Weierstrass Elliptic Functions and the First Integral \texorpdfstring{\eqref{eq8}}{(8)}}
\label{sec:weierstrass}
The first integral \eqref{eq8} shows that the reduced dynamics \eqref{eq7} is a one-degree-of-freedom Hamiltonian system with effective potential $W(\nu)=\nu^6/24-\nu^2/8$, whose central maximum at $\nu=0$ is the saddle of the reduced phase portrait. Since $W(\nu)\to+\infty$ as $|\nu|\to\infty$, every trajectory is bounded: for $-\tfrac1{12}<\nu_0<0$ it is a periodic orbit inside one of the wells around $\nu=\pm1$, and for $\nu_0>0$ a periodic orbit that crosses $\nu=0$. At a generic level the trajectory is not elementary: separation of variables leads to an elliptic integral. The Schuster ($\nu_0=0$) and Srivastava ($\nu_0=\tfrac{1}{12}$) cases are distinguished elementary cases; all other bounded real orbits are naturally expressed through Weierstrass elliptic functions~\cite{Mach2012,MancasRosu2018}.

\medskip
\noindent
\textbf{Elliptic normal form.}
Rewriting \eqref{eq8} as
\begin{equation}\label{eq:elliptic_nu}
\dot{\nu}^2=\frac{1}{12}\bigl(-\nu^6+3\nu^2+C\bigr),
\qquad C:=24\nu_0,
\end{equation}
introduces the Mach parameter $C$ for the autonomous variable $\nu$~\cite{Mach2012}. For $C>0$, the even sextic $-\nu^6+3\nu^2+C$ has exactly two real roots $\pm\nu_{\max}$, and bounded oscillatory motion takes place on $\nu\in[-\nu_{\max},\nu_{\max}]$. Following Mach~\cite{Mach2012}, the quadrature
\begin{equation}\label{eq:nu_quadrature}
t-t_0=\sqrt{12}\int \frac{d\nu}{\sqrt{-\nu^6+3\nu^2+C}}
\end{equation}
is reduced to the Weierstrass form by the substitution $\nu=\pm\sqrt{C/(s-1)}$, which yields
\begin{equation}\label{eq:weierstrass_ode}
t-t_0=\mp\sqrt{12} \int \frac{ds}{\sqrt{4s^3-12s-4(C^2-2)}}.
\end{equation}
The right-hand side is an inverse Weierstrass integral: if $\wp(\tau;g_2,g_3)$ denotes the Weierstrass $\wp$-function with invariants $g_2=12$ and $g_3=4(C^2-2)$, then~\cite{Mach2012}
\begin{equation}\label{eq:wp_identity}
\tau=\int_{\wp(\tau)}^{\infty}\frac{ds}{\sqrt{4s^3-g_2 s-g_3}}.
\end{equation}
The oscillatory solutions of \eqref{eq7} with $C>0$ (equivalently $\nu_0>0$) therefore admit the closed representation
\begin{equation}\label{eq:nu_wp}
\nu(t)=\pm\frac{\sqrt{C}}{\sqrt{\wp\!\left(\tfrac{t-t_0}{2\sqrt{3}};\,12,\,4(C^2-2)\right)-1}},
\end{equation}
with $\wp$ evaluated on the real axis, where $\wp\ge e_1>1$ ($e_1$ the largest real root of $4s^3-12s-4(C^2-2)$); $\nu$ vanishes at the poles of $\wp$ and changes sign there. For $C=1$, \eqref{eq:nu_wp} agrees with a direct DOP853 integration of \eqref{eq7} to $1.9\times10^{-12}$ on $|t-t_0|\le10$ (Figure~\ref{fig:weierstrass_family}). Tracing back through $t=\log\chi$ and $z=z_p(\chi)\nu$ with $z_p(\chi)=\chi^{1/2}/\sqrt{2}$ recovers the Lane--Emden profile in the original variables; this is the Weierstrass family found by Goenner and Havas~\cite{GoennerHavas2000} and systematized by Mach~\cite{Mach2012}.

\medskip
\noindent
\textbf{Energy levels and bifurcation.}
The constant $C$ (equivalently $\nu_0$) plays the role of a bifurcation parameter for \eqref{eq7}:
\begin{enumerate}
\item [i)] $C=0$ ($\nu_0=0$): the separatrix \eqref{eq:schuster_energy} and the Schuster profile; the quadrature \eqref{eq:nu_quadrature} degenerates to elementary functions (Subsection~\ref{sec:schuster}).
\item [ii)] $C=2$ ($\nu_0=\tfrac{1}{12}$): the Srivastava level; \eqref{eq:elliptic_nu} factorizes and yields the closed trigonometric form of Subsection~\ref{sec:srivastava}.
\item [iii)] $C>0$, $C\neq2$: oscillatory orbits with infinitely many zeros in $\xi$; \eqref{eq:nu_wp} applies (the discriminant $g_2^3-27g_3^2=1728-432(C^2-2)^2$ is positive for $0<C<2$ and negative for $C>2$).
\item [iv)] $-2<C<0$: periodic orbits inside one well; $\nu$ keeps its sign, so $\theta=\nu/\sqrt{2\xi}$ has no zeros and is singular at $\xi=0$. At $C=-2$ the orbit reduces to the equilibrium $\nu=\pm1$, i.e.\ the singular solution $\theta=\pm(2\xi)^{-1/2}$.
\end{enumerate}
Thus the classical integrable cases are isolated values of $C$, while the generic motion is periodic in the logarithmic radius $t$.

\medskip
\noindent
\textbf{Connection to Weierstrass methods for related nonlinear equations.}
The reduction \eqref{eq:weierstrass_ode}--\eqref{eq:nu_wp} is structurally parallel to integration techniques developed for other nonlinear evolution equations. Mancas, Spradlin, and Khanal~\cite{MancasBBM2013} showed that traveling-wave reductions of the dissipative Benjamin--Bona--Mahony (BBM) equation become Abel equations whose polynomial Appell invariants yield solutions in terms of $\wp$ and its derivatives, using Ince-type transformations to organize the elliptic moduli. In the same spirit, Mancas and Rosu~\cite{MancasRosu2018} classified two integrable Emden--Fowler classes
\[
z_{\chi\chi}=A\chi^{-\lambda-2}z^n,
\]
including the Lane--Emden case $\lambda=(n-1)/2$ with $n=5$, and wrote parametric Weierstrass solutions for astrophysical models. Their $n=5$ example is precisely the Schuster-reduced setting of Section~\ref{sec:weierstrass}, and places the first integral \eqref{eq8} inside a broader EF integrability framework.

From a dynamical-systems viewpoint, the appearance of $\wp$ is expected: the level sets of \eqref{eq8} are compact curves in the $(\nu,\dot\nu)$-plane for every $\nu_0>-\tfrac1{12}$, and periodic orbits of a Hamiltonian system with a sextic potential generically correspond to elliptic functions (Figure~\ref{fig:weierstrass_family}). The stability analysis in Sections~\ref{sec:stability}--\ref{sec:linear_stability} therefore concerns a profile that lies on a separatrix of this elliptic family rather than on a generic periodic orbit~\eqref{eq:nu_wp}. The Schuster and Weierstrass curves in Figure~\ref{fig:weierstrass_family} are qualitatively different orbit types, not two members of the same family: the Schuster curve ($C=0$) is the homoclinic separatrix, decaying to $\nu=0$ as $t\to\pm\infty$ (Section~\ref{sec:linear_stability}), while the Weierstrass curve ($C=1$) is a periodic orbit, oscillating indefinitely between the turning points $\pm\nu_{\max}$ of the sextic $-\nu^6+3\nu^2+C$. Despite its localized, bell-shaped profile, the Schuster solution is not a soliton in the usual sense: it is a homoclinic solution of the reduced Hamiltonian ODE \eqref{eq7}, not a traveling-wave solution of an integrable evolution equation.

\begin{figure}[htbp]
\centering
\includegraphics[width=0.72\linewidth]{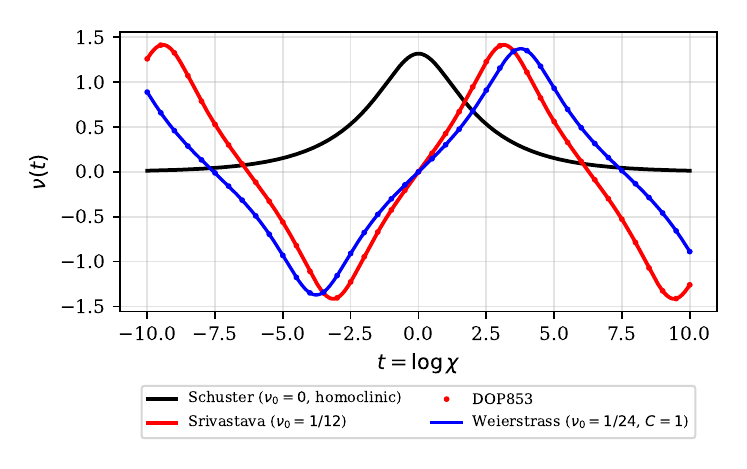}\\[0.6em]
\caption{Reduced solutions $\nu(t)$ for the Schuster separatrix ($\nu_0=0$, black, homoclinic, decaying to $0$ as $t\to\pm\infty$), the Srivastava level ($\nu_0=1/12$, red, smooth and $4\pi$-periodic), and a generic Weierstrass elliptic orbit at $\nu_0=1/24$ ($C=1$, blue, periodic, oscillating between $\pm\nu_{\max}$) from \eqref{eq:nu_wp}. Dots: direct DOP853 integration of \eqref{eq7} (maximum deviations $1.2\times10^{-12}$ and $1.9\times10^{-12}$). Colors match Figures~\ref{fig:energy_levels}, \ref{fig:integrable_nu}, and~\ref{fig:phase_portrait}.}
\label{fig:weierstrass_family}
\end{figure}

\subsection{The Homology-Invariant Phase Plane}
\label{sec:milne}
A second autonomous reduction treats the whole Emden--Fowler family \eqref{eq2} at once. For
\begin{equation}\label{eq:EF_alpha}
\frac{d}{d\xi}\Big(\xi^2\frac{d\eta}{d\xi}\Big)=\alpha\,\xi^{\lambda}\eta^{n},
\end{equation}
which is \eqref{eq2} for $\alpha=-1$ and the Lane--Emden equation \eqref{LE} for $\alpha=-1$, $\lambda=2$, the homology-invariant variables
\begin{equation}\label{eq:milne_vars}
X=\frac{\xi\eta'}{\eta},\qquad Y=\frac{\xi^{\lambda-1}\eta^{n}}{\eta'},\qquad t=\ln\xi,
\end{equation}
($X=-v$, $Y=-u$ in the notation of Chandrasekhar's homology invariants~\cite[Ch.~IV]{Chan}) turn \eqref{eq:EF_alpha} into the planar autonomous system
\begin{equation}\label{eq:milne_system}
\dot X=-X\,(1+X-\alpha Y),\qquad \dot Y=Y\,(1+\lambda+nX-\alpha Y).
\end{equation}
Its equilibria are $(0,0)$, $(-1,0)$, $(0,(1+\lambda)/\alpha)$ and
\begin{equation}\label{eq:milne_X3}
(X_3,Y_3)=\Big(-\frac{\lambda}{n-1},\ \frac{n-1-\lambda}{\alpha(n-1)}\Big).
\end{equation}
The trace $\delta_1$ and determinant $\delta_2$ of the Jacobian, which fix the type of each equilibrium, are listed in Table~\ref{tab:milne}; they depend on $\lambda$ and $n$ only.

\begin{table}[htbp]
\centering
\caption{Equilibria of the homology-invariant system \eqref{eq:milne_system}: trace $\delta_1$ and determinant $\delta_2$ of the Jacobian.}
\label{tab:milne}
\begin{tabular}{ccc}
\hline
Equilibrium & $\delta_1$ & $\delta_2$\\
\hline
$(0,0)$ & $\lambda$ & $-(1+\lambda)$\\
$(-1,0)$ & $2-n+\lambda$ & $1-n+\lambda$\\
$(0,(1+\lambda)/\alpha)$ & $-1$ & $-\lambda(1+\lambda)$\\
$(X_3,Y_3)$ & $\dfrac{1-n+2\lambda}{n-1}$ & $\dfrac{\lambda(n-1-\lambda)}{n-1}$\\
\hline
\end{tabular}
\end{table}

Only $(X_3,Y_3)$ can be a center. Its trace vanishes exactly on the line $n=2\lambda+1$, i.e.\ $\lambda=(n-1)/2$, which is Rosenau's first integrable class~\cite{Rosenau1984} and the first class of~\cite{MancasRosu2018}; there $\delta_2=\lambda/2>0$, so the linearization has purely imaginary eigenvalues. For the Lane--Emden equation ($\lambda=2$) this selects $n=5$: then $(X_3,Y_3)=(-\tfrac12,-\tfrac12)$, which is the singular solution $\theta=(2\xi)^{-1/2}$, and the eigenvalues are $\pm i$. Along $\lambda=2$ the determinant $\delta_2=2(n-3)/(n-1)$ is positive for $n>3$ while the trace $(5-n)/(n-1)$ changes sign at $n=5$: the singular solution is unstable for $3<n<5$ (a node for $n<(11+8\sqrt2)/7\approx3.19$, a focus above) and a stable focus for $n>5$, and $n=5$ is a Hopf-type transition point. Because \eqref{eq7} has the first integral \eqref{eq8}, the transition is degenerate: at $n=5$ the equilibrium is surrounded by a continuous family of closed orbits rather than by a single limit cycle. These closed orbits are the periodic orbits of \eqref{eq7} with $-\tfrac1{12}<\nu_0<0$ around $\nu=\pm1$ (Theorem~\ref{thm:saddle_nu} and the paragraph after it), seen in different variables. A direct DOP853 integration of \eqref{eq:milne_system} with $\alpha=-1$, $n=5$ confirms closed orbits for $\lambda=2$ (the maximal distance from $(X_3,Y_3)$ is the same to $10^{-6}$ on $0\le t\le20$ and $40\le t\le60$), inward spirals for $\lambda=1.8$ and outward spirals for $\lambda=2.2$, as the sign of $\delta_1$ predicts.

\subsection{Complete Classification of Real Solutions for \texorpdfstring{$n=5$}{n=5}: The Work of Mach (2012)}
Mach~\cite{Mach2012} classified all real solutions of the $n=5$ Lane--Emden equation in closed form. This goes beyond the classical work of Schuster, Fowler and Chandrasekhar and puts the known integrable solutions into a single exhaustive framework.

After reduction to the Schuster form \eqref{eq4}, the general real solution can be expressed in terms of Jacobi or Weierstrass elliptic functions. The classification follows from the first integral
\begin{equation}\label{MachFirstIntegral}
\chi z_\chi^2-z\,z_\chi+\frac{1}{3}\chi^{-3}z^6=\nu_0=\frac{C}{24},
\end{equation}
which is \eqref{eq8} written in Schuster's variables through $t=\log\chi$ and $z=z_p(\chi)\nu$ with $z_p(\chi)=\chi^{1/2}/\sqrt{2}$; differentiating and using \eqref{eq4} confirms that it is conserved. In the autonomous variable $\nu$ of \eqref{eq7}, Mach's constant in \eqref{eq:elliptic_nu} satisfies $C=24\nu_0$.

A key result of~\cite{Mach2012} is that the qualitative behavior depends sensitively on the sign and magnitude of $C$ and on the separatrix $C=0$. In particular:
\begin{enumerate}
\item [i)] For $C=0$, the solution reduces to the classical Schuster solution, corresponding to the infinite-radius, finite-mass polytrope with algebraic decay.
\item [ii)] For $C>0$, solutions oscillate in logarithmic radius and admit representations in terms of Jacobi or Weierstrass elliptic functions, with infinitely many zeros in $\xi$; at $C=2$ ($\nu_0=\tfrac{1}{12}$) one recovers the Srivastava branch (Subsection~\ref{sec:srivastava}).
\item [iii)] For $-2\le C<0$, $\nu$ is periodic and of one sign, so $\theta$ is positive for all $\xi>0$, with $\sqrt{\xi}\,\theta$ periodic in $\log\xi$ and $\theta\to\infty$ like $\xi^{-1/2}$ at the center; no real solutions exist for $C<-2$.
\end{enumerate}
Every solution other than the Schuster family is singular at $\xi=0$, since regularity at the center requires $\nu\to0$ as $t\to+\infty$, i.e.\ the separatrix. The Schuster solution thus forms the boundary between positive solutions without zeros ($C<0$) and solutions with infinitely many zeros ($C>0$), reinforcing its role as a critical threshold configuration.

Mach also identified a \emph{discrete scaling symmetry} in the oscillatory solution families. The continuous scaling invariance of the equation at $n=5$ generates the neutral mode of Section~\ref{sec:linear_stability}; the elliptic-function solutions, by contrast, are self-similar only under discrete rescalings of the radial variable. The structure is reminiscent of renormalization-group behavior and is another sign of how special $n=5$ is.

For our purposes Mach’s classification supplies the global picture. The Schuster and Srivastava solutions are non-generic members of a much larger family, and the threshold scaling mode of Section~\ref{sec:linear_stability} is the trace of the continuous scaling symmetry. The Lyapunov instability results are consistent with the existence of nearby elliptic-function solutions that depart qualitatively from the Schuster profile at large radius.

\section{Stability Analysis}
\label{sec:stability}
The distance-varying (nonautonomous) Emden--Fowler model below, its shifted equilibria and its stability theory for general $n$ were introduced by Adams, Mancas and Rosu~\cite{AMR2018}, in the framework of~\cite{HK,SL}. We recall only what is needed and add one result: for $n=5$ we prove instability of the unique equilibrium directly, first for the linearization (Lemma~\ref{lem:pert_growth}, Theorem~\ref{thm:lin_n5_instab}) and then for the full system, with the instability function of~\cite{AMR2018} specialized to $n=5$ and an explicit starting distance $\zeta_0$ (Theorem~\ref{unstable2_n5}).

\begin{remark}
The linearization in Section~\ref{sec:linear_stability} is performed about the explicit profile $\theta_S(\xi)$ in the Lane--Emden boundary-value setting, whereas the linearization in the present section concerns the shifted equilibrium of the distance-varying first-order system in $\zeta$. These are distinct linearized problems and need not share the same spectral character.
\end{remark}
\subsection{Nonautonomous First-Order Formulation and Equilibria}
We study a generalized Emden--Fowler model in the distance variable $\zeta>0$ (shifted so that a finite equilibrium $\bar z=\Omega^{-1/n}$ exists for every $n$ and $\Omega>0$):
\begin{equation}\label{eq:nonauto_general}
\frac{dy}{d\zeta}=f(\zeta,y),\qquad
y(\zeta)=\begin{pmatrix} z(\zeta)\\ z_\zeta(\zeta)\end{pmatrix},
\end{equation}
with
\begin{equation}\label{eq:f_general}
f(\zeta,y)=
\begin{pmatrix}
z_\zeta\\[1mm]
\frac{1}{n+1}\bigl(\Omega z^n-1\bigr)-\frac{2}{\zeta}z_\zeta
\end{pmatrix},
\qquad \zeta\ge \zeta_0>0,
\end{equation}
where $n$ is the polytropic index and $\Omega>0$ is a scaling parameter determined by normalization and boundary conditions. Throughout this section we specialize to the case
\begin{equation}\label{eq:n5}
n=5.
\end{equation}

Equilibria of \eqref{eq:nonauto_general} are distance-invariant points $\bar y$ satisfying $f(\zeta,\bar y)=0$ for all $\zeta$:
\begin{equation}\label{eq:equil_cond}
\bar z_\zeta=0,\qquad \Omega \bar z^n-1=0.
\end{equation}
Hence for odd $n$ (in particular $n=5$) there is a single real equilibrium,
\begin{equation}\label{eq:equil_odd}
\bar y_0^T=\left(\Omega^{-1/n},\,0\right)=\left(\Omega^{-1/5},\,0\right).
\end{equation}
For even $n$ there are two real equilibria $\bar y_{1,2}^T=\left(\mp \Omega^{-1/n},\,0\right)$.

We shift the equilibrium to the origin by setting
\begin{equation}\label{eq:shift}
x^T=y^T-\bar y^T=(x_1,x_2),
\qquad
x_1=z-\bar z,\ \ x_2=z_\zeta.
\end{equation}
For $n=5$ (odd), using $\bar y=\bar y_0$, the shifted system becomes
\begin{equation}\label{system1_n5_plus}
x_\zeta=
\begin{pmatrix}
x_2\\[1mm]
\frac{1}{6}\Bigl[\Omega\left(x_1+\Omega^{-1/5}\right)^5-1\Bigr]-\frac{2}{\zeta}x_2
\end{pmatrix}.
\end{equation}
\paragraph{Results of Adams, Mancas and Rosu~\cite{AMR2018}.} For even $n$ the equilibrium $\bar y_1=(-\Omega^{-1/n},0)$ is uniformly Lyapunov stable and asymptotically stable: the function
\[
V(x)=-\frac{2\Omega}{(n+1)^2}\Bigl[\bigl(x_1-\Omega^{-1/n}\bigr)^{n+1}+\Omega^{-(n+1)/n}\Bigr]+\frac{2}{n+1}x_1+x_2^2
\]
satisfies $V_\zeta=-4x_2^2/\zeta$ along trajectories, it is positive definite on the ball of radius $2\Omega^{-1/n}$ by Bernoulli's inequality, and Barbalat's lemma gives convergence to the equilibrium, with an explicit sublevel-set estimate of the basin. The linearization at $\bar y_1$ also satisfies the linear matrix inequality criterion of Chen and Yang~\cite{CY} with $P(\zeta)=\mathrm{diag}\bigl(1/\zeta,(1+1/n)/(\Omega^{1/n}\zeta)\bigr)$ and $g(\zeta)=-1/\zeta$. For odd $n$ the unique equilibrium is unstable. The remainder of this section treats $n=5$.

\subsection{Linear Instability for \texorpdfstring{$n=5$}{n=5}}
\begin{definition}[Lyapunov and uniform stability]\label{def:lyap}
For the nonautonomous system \eqref{eq:nonauto_general}, the origin is \emph{stable} if for every $\varepsilon>0$ there exists $\delta>0$ such that $\|x(\zeta_0)\|<\delta$ implies $\|x(\zeta)\|<\varepsilon$ for all $\zeta\ge\zeta_0$. It is \emph{uniformly stable} if $\delta$ can be chosen independently of $\zeta_0$, and \emph{unstable} if it is not stable in this sense \cite{HK,SL}.
\end{definition}

We begin by studying the linearization of \eqref{system1_n5_plus} about the origin. Let
\begin{equation}\label{eq:A_def}
A(\zeta)=\nabla_x f(\zeta,x+\bar y)\big|_{x=0}.
\end{equation}
For $n=5$ and $\bar y_0=(\Omega^{-1/5},0)$, differentiation of \eqref{system1_n5_plus} gives
\begin{equation}\label{eq:A_n5}
A(\zeta)=
\begin{pmatrix}
0 & 1\\[1mm]
\frac{5}{6}\Omega^{1/5} & -\frac{2}{\zeta}
\end{pmatrix}.
\end{equation}
Thus the corresponding linear distance-varying system is
\begin{equation}\label{tvls_n5}
x_\zeta=A(\zeta)x,\qquad x(\zeta_0)=x_0.
\end{equation}

\begin{lemma}[Exponential growth of the limiting linearization]\label{lem:lin_growth}
Let
\[
A_0=\begin{pmatrix}0&1\\ \frac{5}{6}\Omega^{1/5}&0\end{pmatrix}.
\]
Then $A_0$ has eigenvalues $\pm\mu$ with $\mu=\sqrt{\frac{5}{6}\Omega^{1/5}}>0$, and there exist $v_+\in\mathbb{R}^2$ and $c>0$ such that every solution of $\frac{d}{d\zeta}x=A_0 x$ with $x(0)=v_+$ satisfies $\|x(\zeta)\|\ge c e^{\mu \zeta}\|v_+\|$ for $\zeta\ge0$.
\end{lemma}

\begin{proof}
The characteristic polynomial is $\lambda^2-\frac{5}{6}\Omega^{1/5}=0$, so $\lambda=\pm\mu$. An eigenvector for $\lambda=\mu$ is $v_+=(1,\mu)^T$, and $e^{A_0 t}v_+=e^{\mu t}v_+$.
\end{proof}

\begin{lemma}[Growth under asymptotically hyperbolic linearization]\label{lem:pert_growth}
Let $A(\zeta)=A_0+B(\zeta)$ with $A_0$ as in Lemma~\ref{lem:lin_growth} and $B(\zeta)=\mathrm{diag}(0,-2/\zeta)$, and let $v_+=(1,\mu)^T$. Then for every $\zeta_0>0$ and every $\delta>0$, the solution of $x_\zeta=A(\zeta)x$ with $x(\zeta_0)=\delta v_+$ is unbounded on $[\zeta_0,\infty)$.
\end{lemma}

\begin{proof}
With $B(\zeta)=\mathrm{diag}(0,-2/\zeta)$, $x_\zeta=A(\zeta)x$ reduces to the scalar second-order equation for $u(\zeta)=x_1(\zeta)$:
\[
u''+\frac{2}{\zeta}u'-\mu^2 u=0,\qquad \mu=\Bigl(\tfrac{5\Omega^{1/5}}{6}\Bigr)^{1/2}>0.
\]
Multiplying by $\zeta^2$ and rearranging gives $(\zeta^2 u')'=\mu^2\zeta^2 u$. Set $\psi(\zeta)=\zeta^2 u'(\zeta)$; then
\[
\psi'(\zeta)=\mu^2\zeta^2 u(\zeta).
\]
For $x(\zeta_0)=\delta v_+=({\delta},\mu\delta)^T$, we have $u(\zeta_0)=\delta>0$ and $u'(\zeta_0)=\mu\delta>0$, so $\psi(\zeta_0)=\zeta_0^2\mu\delta>0$. Since $u'=\psi/\zeta^2>0$ as long as $\psi>0$ and $u>0$, the function $u$ is strictly increasing on $[\zeta_0,\infty)$, giving $u(\zeta)\ge\delta$ for all $\zeta\ge\zeta_0$. Consequently $\psi'(\zeta)=\mu^2\zeta^2 u(\zeta)\ge\mu^2\delta\zeta^2$, so
\[
\psi(\zeta)\ge\zeta_0^2\mu\delta+\tfrac{\mu^2\delta}{3}(\zeta^3-\zeta_0^3).
\]
Dividing by $\zeta^2$, dropping the positive term $\zeta_0^2\mu\delta/\zeta^2$ and integrating gives $u(\zeta)\ge\delta+\tfrac{\mu^2\delta}{6}(\zeta^2-\zeta_0^2)-\tfrac{\mu^2\delta}{3}\zeta_0^2\bigl(1-\tfrac{\zeta_0}{\zeta}\bigr)\ge\tfrac{\mu^2\delta}{6}(\zeta^2-3\zeta_0^2)\to\infty$. Hence $x$ is unbounded. In fact the growth is exponential: since $u''+\tfrac{2}{\zeta}u'=\zeta^{-1}(\zeta u)''$, the equation is $(\zeta u)''=\mu^2\zeta u$, so $u=(a\,e^{\mu\zeta}+b\,e^{-\mu\zeta})/\zeta$, and the data $u(\zeta_0)=\delta$, $u'(\zeta_0)=\mu\delta$ give $a=\delta(1+2\mu\zeta_0)e^{-\mu\zeta_0}/(2\mu)>0$.
\end{proof}

\begin{theorem}\label{thm:lin_n5_instab}
The origin of \eqref{tvls_n5} is unstable for $n=5$ when shifting about $\bar y_0=(\Omega^{-1/5},0)$.
\end{theorem}

\begin{proof}
Apply Lemma~\ref{lem:pert_growth} with $B(\zeta)=\begin{pmatrix}0&0\\ 0&-2/\zeta\end{pmatrix}$ and $A_0$ from \eqref{eq:A_n5}. Boundedness of the origin would contradict Lemma~\ref{lem:pert_growth} for any $\delta>0$ and $\zeta_0$, so the origin is unstable.
\end{proof}

\begin{remark}
Theorem~\ref{thm:lin_n5_instab} is the $n=5$ specialization of the general odd-$n$ instability phenomenon: the unique real equilibrium $\bar y_0$ lies on a branch whose linearization has a positive stiffness term, producing saddle-type behavior in the absence of strong damping. The $-2/\zeta$ decay term does not prevent instability because it vanishes as $\zeta\to\infty$.
\end{remark}
\subsection{Nonlinear Instability for \texorpdfstring{$n=5$}{n=5}}

We recall the instability criterion (cf.\ \cite{SL}):

\begin{theorem}\emph{(\cite{SL})}\label{unstable_thm}
If, in a neighborhood $\mathcal{B}_0$ of the origin, there exists a continuously differentiable, decrescent function $V(x,\zeta)$ such that
\begin{enumerate}
\item [i)] $V(0,\zeta)=0$ for all $\zeta\ge\zeta_0$,
\item [ii)] $V(x,\zeta_0)$ assumes strictly positive values arbitrarily close to the origin,
\item [iii)] $V_\zeta(x,\zeta)$ is locally positive definite in $\mathcal{B}_0$,
\end{enumerate}
then the equilibrium is unstable.
\end{theorem}

\begin{theorem}\label{unstable2_n5}
The equilibrium point $x\equiv 0$ of the shifted $n=5$ system \eqref{system1_n5_plus} is unstable.
\end{theorem}

\begin{proof}
Define the scalar function (the function of~\cite{AMR2018} for odd $n$, at $n=5$)
\begin{equation}\label{V_instab_n5}
V(x,\zeta)=\Bigl(\Omega(x_1+\Omega^{-1/5})^5-1\Bigr)x_2+\frac{6x_2^2}{\zeta}.
\end{equation}
Clearly $V(0,\zeta)=0$. For $x_1=0$ and $x_2\neq 0$, $V(x,\zeta)=\frac{6x_2^2}{\zeta}>0$, so condition (ii) holds. Differentiating $V$ along trajectories of \eqref{system1_n5_plus} yields
\begin{align}
V_\zeta
&=\frac{1}{6}\Bigl(\Omega(x_1+\Omega^{-1/5})^5-1\Bigr)^2
+5\Omega (x_1+\Omega^{-1/5})^4 x_2^2
-\frac{30x_2^2}{\zeta^2}. \label{Vdot_instab_n5}
\end{align}
Fix $r>0$ sufficiently small so that $x_1+\Omega^{-1/5}\ge \frac12\Omega^{-1/5}$ in $B_r(0)$. Then
\[
5\Omega (x_1+\Omega^{-1/5})^4\ge 5\Omega\left(\tfrac12\Omega^{-1/5}\right)^4=\frac{5}{16}\Omega^{1/5}.
\]
Choose $\zeta_0$ large enough such that $\frac{30}{\zeta_0^2}\le \frac{5}{32}\Omega^{1/5}$, i.e.
\begin{equation}\label{zeta0_choice_n5}
\zeta_0^2\ge \frac{192}{\Omega^{1/5}}.
\end{equation}
Then for $\zeta\ge\zeta_0$ and $\|x\|<r$,
\[
5\Omega (x_1+\Omega^{-1/5})^4-\frac{30}{\zeta^2}\ge \frac{5}{32}\Omega^{1/5}>0,
\]
and hence \eqref{Vdot_instab_n5} implies $V_\zeta\ge V_0(x)$, where
\begin{equation}
V_0(x)=\frac{1}{6}\Bigl(\Omega(x_1+\Omega^{-1/5})^5-1\Bigr)^2+\frac{5}{32}\Omega^{1/5}x_2^2
\end{equation}
is locally positive definite. Therefore the conditions of Theorem~\ref{unstable_thm} hold and $x\equiv 0$ is unstable.
\end{proof}

\subsection{Invariant Manifold Structure}
Consider the autonomous reduction
\begin{equation}
\ddot{\nu}-\frac14(\nu-\nu^5)=0.
\end{equation}

\begin{theorem}[Hyperbolic saddle in the reduced phase plane]\label{thm:saddle_nu}
The equilibrium $\nu=0$ of \eqref{eq7} is a hyperbolic saddle with one-dimensional stable and unstable manifolds.
\end{theorem}

\begin{proof}
Linearizing \eqref{eq7} at $\nu=0$ gives $\ddot{\eta}-\eta/4=0$, whose first-order form has Jacobian eigenvalues $\pm\tfrac12$. By the stable manifold theorem, there exist one-dimensional invariant manifolds, as smooth as the vector field, tangent to the corresponding eigendirections; the Schuster profile \eqref{eq:nu_simplified} lies on both of them. More precisely, since $\nu_{\rm Schuster}(t)=3^{1/4}\sqrt{\sech(t-t_0)}\to0$ as $t\to\pm\infty$, this orbit departs along the unstable manifold of $\nu=0$ and returns along its stable manifold, forming a homoclinic orbit to the saddle. The homoclinic structure is the geometric expression of the separatrix role of the Schuster solution in the global phase portrait (Figure~\ref{fig:phase_portrait}).
\end{proof}

The remaining equilibria of \eqref{eq7} are $\nu=\pm1$, the minima of $W$ at the level $\nu_0=-\tfrac1{12}$. Linearizing $\ddot\nu=\tfrac14(\nu-\nu^5)$ at $\nu_c$ gives $\ddot\eta=\tfrac14(1-5\nu_c^4)\eta$, with eigenvalues $\pm\tfrac12\sqrt{1-5\nu_c^4}$: $\pm\tfrac12$ at $\nu_c=0$ (the saddle) and $\pm i$ at $\nu_c=\pm1$. Since \eqref{eq7} is conservative, $\nu=\pm1$ are nonlinear centers, surrounded by the periodic orbits with $-\tfrac1{12}<\nu_0<0$; they correspond to the singular solutions $\theta=\pm(2\xi)^{-1/2}$.

\begin{figure}[htbp]
\centering
\includegraphics[width=0.78\linewidth]{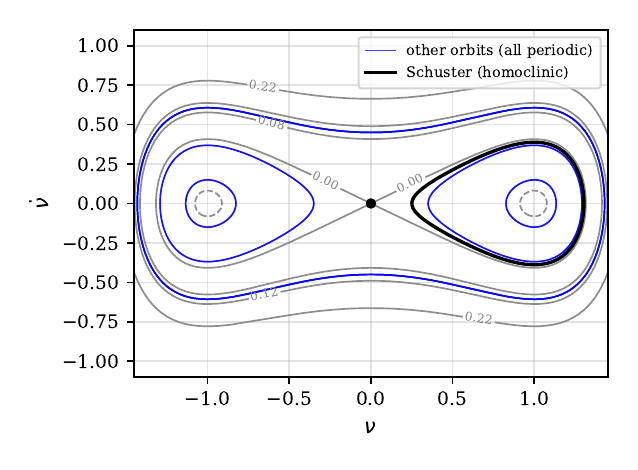}\\[0.6em]
\caption{Phase portrait of the reduced system \eqref{eq7} with level sets of the first integral (gray contours) and numerically integrated trajectories. The equilibrium at $\nu=0$ (black dot) is a hyperbolic saddle. The black curve is the Schuster solution, the homoclinic orbit that departs along the unstable manifold and returns along the stable manifold, forming the separatrix loops; it is color-matched to the Schuster curve in Figures~\ref{fig:energy_levels}, \ref{fig:integrable_nu}, and~\ref{fig:weierstrass_family}. Blue curves are other representative trajectories, all periodic since $W(\nu)$ grows like $\nu^6$: closed loops around $\nu=\pm1$ inside the separatrix ($-\tfrac1{12}<\nu_0<0$) and closed curves encircling both loops outside it ($\nu_0>0$), the type of the Weierstrass curve of Figure~\ref{fig:weierstrass_family}.}
\label{fig:phase_portrait}
\end{figure}

\begin{remark}[Relation to Section~\ref{sec:stability}]
The saddle in Theorem~\ref{thm:saddle_nu} governs the autonomous $\nu$-dynamics after Schuster reduction, whereas Theorems~\ref{thm:lin_n5_instab} and~\ref{unstable2_n5} concern the distance-varying system \eqref{system1_n5_plus} in the variable $\zeta$, an auxiliary model whose nonlinear term $+\Omega z^n/(n+1)$ has the opposite sign of the Lane--Emden term $-z^n$. Both mechanisms reinforce the non-attracting character of the $n=5$ odd-index equilibrium.
\end{remark}

\section{Variational Structure}
\label{sec:variational}
\begin{proposition}[Euler--Lagrange structure at $n=5$]\label{prop:euler_lagrange}
Formal variation of
\begin{equation}\label{eq:action_theta}
\mathcal{J}[\theta]
=
\int_0^\infty
\left(
\frac12 \theta_\xi^2
-\frac16 \theta^6
\right)\xi^2\,d\xi
\end{equation}
yields the Lane--Emden equation \eqref{eq:lane_emden_n5}. The Schuster solution $\theta_S$ is a critical point, and the second variation at $\theta_S$ equals the quadratic form $Q[\phi]=\langle L\phi,\phi\rangle$ of the operator $L$ in Section~\ref{sec:linear_stability} (equivalently, twice the energy functional \eqref{eq:energy_functional}).
\end{proposition}

\begin{proof}
For a smooth, compactly supported perturbation $\phi$,
\[
\delta\mathcal{J}[\theta](\phi)
=
\int_0^\infty
\left(
\theta_\xi\phi_\xi-\theta^5\phi
\right)\xi^2\,d\xi.
\]
Integrating the kinetic term by parts (the boundary term $\xi^2\theta_\xi\phi$ vanishes at $\xi=0$ because of the factor $\xi^2$ and at infinity because $\phi$ has compact support) gives
\[
\delta\mathcal{J}[\theta](\phi)
=
-\int_0^\infty
\left(
\frac{1}{\xi^2}\frac{d}{d\xi}\!\left(\xi^2\theta_\xi\right)+\theta^5
\right)\phi\,\xi^2\,d\xi,
\]
which vanishes for all admissible $\phi$ if and only if \eqref{eq:lane_emden_n5} holds. The second variation at $\theta_S$ is
\[
\delta^2\mathcal{J}[\theta_S](\phi)
=\int_0^\infty\left(\phi_\xi^2-5\theta_S^4\phi^2\right)\xi^2\,d\xi
=Q[\phi]=2\,\mathcal{E}[\phi],
\]
which is twice the quadratic form \eqref{eq:energy_functional} associated with $L$.
\end{proof}

\section{Linear Stability of the \texorpdfstring{$n=5$}{n=5} Lane--Emden Solutions}
\label{sec:linear_stability}
We analyze the linear stability of the $n=5$ Lane--Emden equation by studying small perturbations about the Schuster solution. This analysis complements the nonlinear instability results of Section~\ref{sec:stability} and the saddle structure of Theorem~\ref{thm:saddle_nu}.

\begin{remark}[Two linearizations]\label{rem:two_linearizations}
Section~\ref{sec:stability} linearizes the distance-varying system \eqref{system1_n5_plus} about the constant equilibrium $(\Omega^{-1/5},0)$ and establishes exponential instability (Theorems~\ref{thm:lin_n5_instab} and~\ref{unstable2_n5}). The present section linearizes about the \emph{boundary-value} profile $\theta_S$ in the original variable $\xi$. Proposition~\ref{prop:scaling} and Remark~\ref{rem:phi0_bc} show that this problem admits the neutral tangent mode $\phi_0=\tfrac12\theta_S+\xi\theta_S'$ at the threshold $0$ of $\sigma_{\mathrm{ess}}(\overline{L})$; since $\phi_0\notin\mathcal{H}$, $0$ is not an $L^2$ eigenvalue. The boundary-value operator $L=-\Delta_r-5\theta_S^4$ additionally carries a genuine negative-energy bound state (Theorem~\ref{thm:marginal_spectral}), so it too is spectrally unstable. The hyperbolic saddle in the $\nu$-variables and the negative spectrum in $\xi$ are different manifestations of the same critical index $n=5$.
\end{remark}
\subsection{Perturbation about the Schuster Solution}
Let $\theta_S(\xi)$ denote the Schuster solution,
\begin{equation}\label{eq:schuster_solution}
\theta_S(\xi)=\left(1+\frac{\xi^2}{3}\right)^{-1/2},
\end{equation}
which satisfies the Lane--Emden equation
\begin{equation}\label{eq:lane_emden_n5}
\frac{1}{\xi^2}\frac{d}{d\xi}\left(\xi^2\frac{d\theta}{d\xi}\right)+\theta^5=0.
\end{equation}
We consider a perturbed solution of the form
\begin{equation}\label{eq:perturbation_ansatz}
\theta(\xi)=\theta_S(\xi)+\varepsilon\,\phi(\xi),\qquad 0<\varepsilon\ll1,
\end{equation}
and substitute \eqref{eq:perturbation_ansatz} into \eqref{eq:lane_emden_n5}. Expanding
\begin{equation}
(\theta_S+\varepsilon\phi)^5
=\theta_S^5+5\varepsilon\,\theta_S^4\phi+O(\varepsilon^2).
\end{equation}
Since $\theta_S$ satisfies \eqref{eq:lane_emden_n5}, the $O(\varepsilon)$ terms give the linearized equation
\begin{equation}\label{eq:linearized_perturbation}
\frac{1}{\xi^2}\frac{d}{d\xi}\left(\xi^2\frac{d\phi}{d\xi}\right)
+5\theta_S(\xi)^4\,\phi=0.
\end{equation}
Using $\theta_S(\xi)^4=(1+\xi^2/3)^{-2}$, equation \eqref{eq:linearized_perturbation} becomes
\begin{equation}\label{eq:sturm_liouville_form}
\frac{1}{\xi^2}\frac{d}{d\xi}\left(\xi^2\phi_\xi\right)
+5\left(1+\frac{\xi^2}{3}\right)^{-2}\phi=0.
\end{equation}
\subsection{Self-Adjoint Structure}
The linearized equation about the Schuster solution takes the form
\begin{equation}\label{eq:self_adjoint_form}
\left(\xi^2\phi_\xi\right)_\xi
+5\xi^2\left(1+\frac{\xi^2}{3}\right)^{-2}\phi=0.
\end{equation}
This equation is of Sturm--Liouville type
\begin{equation}\label{eq:SL_general}
-(p(\xi)\phi_\xi)_\xi + q(\xi)\phi = 0,
\end{equation}
with
\[
p(\xi)=\xi^2,
\qquad
q(\xi)=-5\xi^2\left(1+\frac{\xi^2}{3}\right)^{-2}.
\]
Both coefficients are smooth for $\xi>0$, with $p>0$ and $q<0$; the sign of $q$ is the one used in the normal form \eqref{eq:SL_normal} below.

\medskip
\noindent
\textbf{Weighted Hilbert space formulation.}
Equation \eqref{eq:self_adjoint_form} defines a symmetric differential operator
\begin{equation}\label{eq:operator_def}
\mathcal{L}\phi
=
-\frac{1}{\xi^2}\frac{d}{d\xi}\!\left(\xi^2\frac{d\phi}{d\xi}\right)
-5\left(1+\frac{\xi^2}{3}\right)^{-2}\phi,
\end{equation}
acting on the weighted space
\[
L^2\big((0,\infty),\xi^2 d\xi\big),
\]
which corresponds to the natural radial measure in three dimensions.
Integration by parts shows that for sufficiently regular functions
$\phi,\psi$ with compact support,
\begin{equation}\label{eq:symmetry}
\int_0^\infty \xi^2 \psi\, \mathcal{L}\phi\, d\xi
=
\int_0^\infty \xi^2 \phi\, \mathcal{L}\psi\, d\xi,
\end{equation}
so $\mathcal{L}$ is formally self-adjoint with respect to the inner product
\[
\langle \phi,\psi\rangle
=
\int_0^\infty \xi^2 \phi(\xi)\psi(\xi)\,d\xi.
\]

\medskip
\noindent
\textbf{Conservation of the Wronskian.}
Let $\phi_1$ and $\phi_2$ be two solutions of
\eqref{eq:self_adjoint_form}. Their Wronskian is
\begin{equation}\label{eq:wronskian}
W(\phi_1,\phi_2)
=
\xi^2\bigl(\phi_1\phi_{2,\xi}-\phi_2\phi_{1,\xi}\bigr).
\end{equation}
Differentiating and using \eqref{eq:self_adjoint_form}, we obtain
\[
\frac{dW}{d\xi}
=
\xi^2(\phi_1\phi_{2,\xi\xi}-\phi_2\phi_{1,\xi\xi})
+2\xi(\phi_1\phi_{2,\xi}-\phi_2\phi_{1,\xi}),
\]
and substituting
\[
\phi_{\xi\xi}
=
-\frac{2}{\xi}\phi_\xi
-5\left(1+\frac{\xi^2}{3}\right)^{-2}\phi,
\]
all terms cancel identically, giving
\[
\frac{dW}{d\xi}=0.
\]
Thus $W(\phi_1,\phi_2)$ is constant on $(0,\infty)$.

This conservation law ensures that any two linearly independent solutions form a global fundamental system and permits the classification of oscillatory versus non-oscillatory behavior via Sturm comparison theory.

\medskip
\noindent
\textbf{Variational structure.}
Equation \eqref{eq:self_adjoint_form} arises as the Euler--Lagrange equation of the quadratic functional
\begin{equation}\label{eq:energy_functional}
\mathcal{E}[\phi]
=
\frac{1}{2}
\int_0^\infty
\left(
\xi^2 \phi_\xi^2
-
5\xi^2\left(1+\frac{\xi^2}{3}\right)^{-2}\phi^2
\right)d\xi.
\end{equation}
The linearized dynamics are therefore governed by an energy principle: the sign structure of $\mathcal{E}$ determines whether perturbations can grow, and the existence of a zero mode reflects degeneracy of this quadratic form.

\medskip
\noindent
\textbf{Endpoint classification.}
At $\xi=0$, regularity of radial solutions imposes
\[
\phi(\xi)=\phi(0)+O(\xi^2),
\qquad
\phi_\xi(0)=0.
\]
At infinity, since
\[
q(\xi)=-5\xi^2(1+\xi^2/3)^{-2}\sim -45\,\xi^{-2},
\qquad \xi\to\infty,
\]
the coefficient $q(\xi)/p(\xi)=q(\xi)/\xi^2\sim -45\,\xi^{-4}\to0$ (we divide the Sturm--Liouville form \eqref{eq:SL_general} through by the leading coefficient $p(\xi)=\xi^2$, as required for the Weyl limit-point analysis), so the equation reduces asymptotically to
\[
\phi_{\xi\xi}+\frac{2}{\xi}\phi_\xi=0,
\]
whose solutions behave as
\[
\phi(\xi)\sim C_1 + \frac{C_2}{\xi}.
\]
Thus $\xi=\infty$ is of limit-point type, and no additional boundary condition is required there for self-adjointness.

\medskip
\noindent
\textbf{Spectral implications.}
Since $\mathcal{L}$ is formally self-adjoint and the potential $5\theta_S^4$ decays as $O(\xi^{-4})$, the Friedrichs extension has essential spectrum $[0,\infty)$ (Theorem~\ref{thm:ess_spec}). The scaling generator $\phi_0$ satisfies $\mathcal{L}\phi_0=0$ but is not square-integrable with weight $\xi^2$; thus $0$ is a threshold of the essential spectrum rather than an $L^2$ eigenvalue. Because the potential $-5\theta_S^4$ is attractive, the operator additionally admits a negative-energy bound state (Theorem~\ref{thm:marginal_spectral}). This spectral structure is the spectral instability of the $n=5$ Schuster configuration.

\subsection{Spectral Formulation of the Linearized Operator}
To place the linear stability problem in a functional-analytic framework, we define the linearized operator (identical to $\mathcal{L}$ in \eqref{eq:operator_def})
\begin{equation}
L\phi
:=
-\frac{1}{\xi^2}\frac{d}{d\xi}\left(\xi^2\frac{d\phi}{d\xi}\right)
-5\theta_S(\xi)^4\phi,
\label{def_L}
\end{equation}
acting on the Hilbert space
\begin{equation}
\mathcal{H}=L^2((0,\infty),\xi^2\,d\xi).
\end{equation}

Define the quadratic form
\begin{equation}\label{eq:quadratic_form}
Q[\phi]
=
\int_0^\infty
\left(
\phi_\xi^2-5\theta_S(\xi)^4\phi^2
\right)\xi^2\,d\xi,
\qquad
\phi\in C_c^\infty([0,\infty)),\ \phi'(0)=0.
\end{equation}
Let $D_{\mathrm{core}}$ be the completion of $C_c^\infty([0,\infty))$ with respect to the norm $\|\phi\|_{H}^2=Q[\phi]+\|\phi\|_{\mathcal{H}}^2$.

\begin{lemma}[Sturm--Liouville normal form]\label{lem:SL_normal}
With $p(\xi)=\xi^2$, $q(\xi)=-5\xi^2\theta_S(\xi)^4<0$, and $w(\xi)=\xi^2$, the eigenvalue problem $L\phi=\lambda\phi$ for the operator \eqref{eq:operator_def} is equivalent to
\begin{equation}\label{eq:SL_normal}
-(p(\xi)\phi')' + q(\xi)\phi = \lambda\,w(\xi)\phi,
\end{equation}
which is a singular Sturm--Liouville problem on $(0,\infty)$: the endpoint $\xi=0$ is singular ($p(0)=0$) and of limit-circle type, since both local solutions $\phi\sim1$ and $\phi\sim\xi^{-1}$ are square-integrable near $0$ with weight $\xi^2$, so a boundary condition is needed there (the Friedrichs extension selects the regular solution, $\phi$ bounded, i.e.\ $\xi\phi\to0$); the endpoint $\xi=\infty$ is of limit-point type \cite[Ch.~10]{Zettl}. The homogeneous linearized equation \eqref{eq:sturm_liouville_form} is the case $\lambda=0$.
\end{lemma}

\begin{lemma}[Rayleigh quotient]\label{lem:rayleigh_nonneg}
Let $p,q,w$ be as in Lemma~\ref{lem:SL_normal}. If $\phi\in D_{\mathrm{core}}$ satisfies \eqref{eq:SL_normal} for some $\lambda\in\mathbb{R}$, then
\begin{equation}\label{eq:rayleigh}
\lambda
=
\frac{\displaystyle\int_0^\infty p(\xi)\,\phi_\xi(\xi)^2\,d\xi
-5\displaystyle\int_0^\infty \xi^2\theta_S(\xi)^4\,\phi(\xi)^2\,d\xi}
{\|\phi\|_{\mathcal{H}}^2}
=\frac{Q[\phi]}{\|\phi\|_{\mathcal{H}}^2}.
\end{equation}
The numerator is the quadratic form $Q[\phi]$ of \eqref{eq:quadratic_form}, which is \emph{not} sign-definite: the attractive potential $-5\theta_S^4$ permits trial functions with $Q[\phi]<0$. Consequently, $\overline{L}$ possesses negative spectrum (see Remark~\ref{rem:negative_spectrum}).
\end{lemma}

\begin{proof}
Multiply \eqref{eq:SL_normal} by $\phi$ and integrate on $(0,\infty)$. Using
\[
\int_0^\infty (p\phi')'\phi
=
\Bigl[p(\xi)\phi(\xi)\phi'(\xi)\Bigr]_{0}^{\infty}
-\int_0^\infty p(\xi)\phi'(\xi)^2\,d\xi
=-\int_0^\infty p(\xi)\phi'(\xi)^2\,d\xi
\]
for $\phi\in D_{\mathrm{core}}$. Multiplying \eqref{eq:SL_normal} by $\phi$ and integrating gives
\[
-\int_0^\infty (p\phi')'\phi\,d\xi
+\int_0^\infty q\phi^2\,d\xi
=\lambda\int_0^\infty w\phi^2\,d\xi,
\]
whence \eqref{eq:rayleigh} follows. The first numerator integral is nonnegative, but the second (with $q=-5\xi^2\theta_S^4<0$) is strictly negative, so the sign of $\lambda$ is not determined a priori; explicit trial functions yield $Q[\phi]<0$.
\end{proof}

\begin{remark}[Existence of negative spectrum]\label{rem:negative_spectrum}
Since $\sigma_{\mathrm{ess}}(\overline{L})=[0,\infty)$ (Theorem~\ref{thm:ess_spec}) while there exist normalized $\phi\in D_{\mathrm{core}}$ with $Q[\phi]<0$ (e.g., Gaussians $\phi=e^{-(\xi/a)^2}$, which satisfy $\phi'(0)=0$, with $0.92\lesssim a\lesssim 37$; for larger $a$ the kinetic term, which grows like $a$, dominates the bounded potential term), the min--max principle yields $\inf\sigma(\overline{L})<0$---at least one negative eigenvalue (bound state) below the essential spectrum. This is consistent with Sturm oscillation theory: the zero-energy solution $\phi_0$ of $L\phi_0=0$ has a single node at $\xi=\sqrt{3}$, so a nodeless ground state of strictly negative energy must lie below it. Theorem~\ref{thm:marginal_spectral}(ii) shows that there is exactly one; numerically (finite differences for $u=\xi\phi$ on $[0,400]$) $\lambda_0\approx-1.2104$, and the next eigenvalue tends to $0$ as the interval grows.
\end{remark}

\begin{proposition}[Symmetry and Friedrichs extension]\label{prop:symmetric_L}
The densely defined operator $L$ with domain $D_{\mathrm{core}}$ is symmetric on $\mathcal{H}$. Since the potential $-5\theta_S^4$ is bounded ($0<5\theta_S^4\le5$), the form $Q$ is bounded below by $-5\|\phi\|_{\mathcal{H}}^2$. The Friedrichs extension $\overline{L}$ is therefore self-adjoint with $\inf\sigma(\overline{L})\ge-5$ and form domain
\[
\mathcal{D}(Q)=\{\phi\in\mathcal{H}:\phi \text{ abs.\ cont.},\ \phi'\in\mathcal{H},\ Q[\phi]<\infty\}.
\]
\end{proposition}

\begin{proof}
For $\phi,\psi\in D_{\mathrm{core}}$, integration by parts gives
\begin{equation}\label{eq:symmetric_identity}
\langle L\phi,\psi\rangle
=
\int_0^\infty
\left(
\phi_\xi \psi_\xi
-5\theta_S(\xi)^4 \phi\psi
\right)\xi^2 d\xi
=\tfrac12\bigl(Q[\phi+\psi]-Q[\phi]-Q[\psi]\bigr),
\end{equation}
confirming symmetry. The Friedrichs extension exists because $Q$ is closable and bounded below ($Q[\phi]\ge-5\|\phi\|_{\mathcal{H}}^2$ since $0<5\theta_S^4\le5$) on $D_{\mathrm{core}}$; see \cite[Sec.~VI.3]{ReedSimon4}.
\end{proof}

Thus the linear stability problem reduces to spectral analysis of $L$:
\begin{equation}
L\phi=\lambda\phi.
\end{equation}
\subsection{Essential Spectrum}
As $\xi\to\infty$,
\[
L\sim -\Delta_r,
\]
the radial Laplacian in $\mathbb{R}^3$.

\begin{theorem}[Essential spectrum]\label{thm:ess_spec}
For the Friedrichs extension $\overline{L}$,
\[
\sigma_{\mathrm{ess}}(\overline{L})=[0,\infty).
\]
\end{theorem}

\begin{proof}
Write $\overline{L}=-\frac{1}{\xi^2}\frac{d}{d\xi}\left(\xi^2\frac{d}{d\xi}\right)+V(\xi)$ with $V(\xi)=-5\theta_S(\xi)^4=O(\xi^{-4})$ as $\xi\to\infty$. The potential $V$ is bounded and continuous on $(0,\infty)$ and vanishes at infinity, so $V$ is relatively compact with respect to the free operator $L_0=-\xi^{-2}(\xi^2\phi')'$: the resolvent $(L_0+1)^{-1}$ maps into decaying functions, and multiplication by a potential vanishing at infinity composed with this resolvent is compact. By Weyl's theorem on the invariance of the essential spectrum under relatively compact perturbations, $\sigma_{\mathrm{ess}}(\overline{L})$ equals that of $L_0$ on the radial domain with regularity at the origin. For $L_0$, the essential spectrum is $[0,\infty)$ (limit-point case at $\xi=\infty$ with plane-wave asymptotics $\phi\sim e^{\pm i\sqrt{\lambda}\,\xi}/\xi$ for $\lambda\ge0$).
\end{proof}

Thus $\lambda=0$ lies at the threshold of the essential spectrum, consistent with the neutral, non-exponential character of the scaling mode.
\subsection{Regularity and Boundary Conditions}
Physically admissible perturbations must be regular at the origin. A Frobenius expansion of \eqref{eq:sturm_liouville_form} near $\xi=0$ yields
\begin{equation}
\phi(\xi)=\phi(0)+O(\xi^2),\qquad \phi_\xi(0)=0,
\end{equation}
so regular perturbations are uniquely determined by their central value.

For large $\xi$, the coefficient satisfies
\begin{equation}
5\left(1+\frac{\xi^2}{3}\right)^{-2}
\sim 45\,\xi^{-4},\qquad \xi\to\infty,
\end{equation}
and \eqref{eq:sturm_liouville_form} asymptotically reduces to
\begin{equation}\label{eq:asymptotic_linearized}
\phi_{\xi\xi}+\frac{2}{\xi}\phi_\xi+O(\xi^{-4})\phi=0.
\end{equation}
The leading-order solutions are
\begin{equation}\label{eq:asymptotic_solutions}
\phi(\xi)\sim C_1+\frac{C_2}{\xi},\qquad \xi\to\infty.
\end{equation}
Generic linear perturbations therefore do not decay faster than $\xi^{-1}$; a non-decaying component persists unless excluded by boundary conditions. This slow decay reflects the infinite extent of the $n=5$ polytrope.
\subsection{Neutral Mode and Scaling Invariance}
The $n=5$ Lane--Emden equation possesses a continuous scaling symmetry: if $\theta(\xi)$ is a solution, then so is
\begin{equation}\label{eq:scaling_symmetry}
\theta_\lambda(\xi)=\lambda^{1/2}\theta(\lambda\xi),
\qquad \lambda>0.
\end{equation}
Differentiating \eqref{eq:scaling_symmetry} with respect to $\lambda$ at $\lambda=1$ produces a distinguished perturbation mode,
\begin{equation}\label{eq:neutral_scaling_mode}
\phi_0(\xi)
=\left.\frac{d}{d\lambda}\theta_\lambda(\xi)\right|_{\lambda=1}
=\frac{1}{2}\theta_S(\xi)+\xi\,\theta_S'(\xi).
\end{equation}
\begin{lemma}[Neutral mode]\label{lem:neutral_mode}
The function $\phi_0$ in \eqref{eq:neutral_scaling_mode} satisfies $L\phi_0=0$ on $(0,\infty)$. As $\xi\to\infty$, $\phi_0(\xi)\sim C/\xi$ for some $C\neq0$, so $\phi_0\notin\mathcal{H}$. The quadratic form $Q[\phi_0]$ is finite (the integrand in \eqref{eq:quadratic_form} decays like $O(\xi^{-2})$).
\end{lemma}

\begin{proof}
Differentiating \eqref{eq:scaling_symmetry} at $\lambda=1$ gives $\phi_0=\tfrac12\theta_S+\xi\theta_S'$. Since $\theta_\lambda$ solves \eqref{eq:lane_emden_n5} for each $\lambda$, differentiating the equation with respect to $\lambda$ at $\lambda=1$ yields $L\phi_0=0$. With $\theta_S(\xi)\sim\sqrt{3}\,\xi^{-1}$, one has $\phi_0(\xi)\sim-\tfrac{\sqrt{3}}{2}\xi^{-1}$, hence $\|\phi_0\|_{\mathcal{H}}^2=\int_0^\infty \xi^2\phi_0^2\,d\xi$ diverges. For $Q[\phi_0]$, the term $\xi^2\phi_\xi^2$ is $O(\xi^{-2})$ and $5\xi^2\theta_S^4\phi_0^2$ is $O(\xi^{-4})$, so the integral converges.
\end{proof}

The mode $\phi_0$ neither grows nor decays exponentially; it represents infinitesimal rescalings of the background profile. Since $\phi_0\notin\mathcal{H}$, it is not an eigenfunction of $\overline{L}$ in the Hilbert-space sense; rather, $0$ is a threshold of $\sigma_{\mathrm{ess}}(\overline{L})$ (Theorem~\ref{thm:ess_spec}; see Figure~\ref{fig:linear_operator}).

\begin{remark}[Neutral mode and boundary data]\label{rem:phi0_bc}
The neutral mode satisfies $\phi_0(0)=\tfrac12\theta_S(0)=\tfrac12$, so it is not an admissible perturbation under homogeneous central data $\phi(0)=0$. Rather, $\phi_0$ is the tangent vector to the scaling orbit \eqref{eq:scaling_symmetry} at $\lambda=1$: it lies in the kernel of the formal operator $L$ and encodes neutral (non-exponential) linear motion along the similarity direction. Physical perturbations with prescribed $\phi(0)$ and $\phi'(0)=0$ are analyzed in the Friedrichs framework of Section~\ref{sec:linear_stability}; Theorem~\ref{thm:marginal_spectral} shows that $0$ is the bottom of the essential spectrum but that a negative eigenvalue exists below it.
\end{remark}

\begin{figure}[htbp]
\centering
\includegraphics[width=0.72\linewidth]{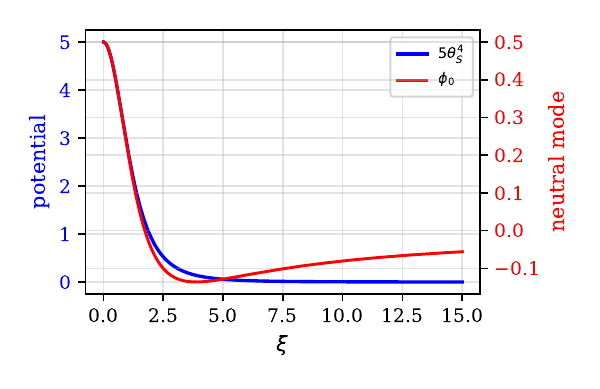}\\[0.6em]
\caption{Linearized potential $5\theta_S(\xi)^4$ (blue, left axis) and neutral mode $\phi_0=\tfrac12\theta_S+\xi\theta_S'$ (red, right axis) associated with scaling symmetry at $n=5$.}
\label{fig:linear_operator}
\end{figure}
\subsection{Interpretation and Connection to Nonlinear Results}
The linear analysis leads to the following conclusions:
\begin{enumerate}
\item [i)] The linearized operator about the Schuster solution is a Schr\"odinger operator with an attractive potential well. It possesses a negative-energy bound state, so the second-variation form $Q$ is indefinite---the Schuster profile is a saddle point of the second-variation energy, with a direction along which the energy strictly decreases, and is therefore spectrally unstable.
\item [ii)] The formal neutral mode \eqref{eq:neutral_scaling_mode} sits at the threshold $0$ of the essential spectrum and encodes the neutral scaling direction; it is not square-integrable and does not by itself confer stability.
\item [iii)] Generic perturbations exhibit slow algebraic behavior at large radius, consistent with the infinite extent of the $n=5$ polytrope.
\end{enumerate}

These linear results are consistent with the instability of the distance-varying model of Section~\ref{sec:stability} and with the saddle structure of Theorem~\ref{thm:saddle_nu}: perturbations that leave the separatrix in the $(\nu,\dot\nu)$-plane lie on periodic orbits and give solutions that differ qualitatively from $\theta_S$ at large radius. Together, these results identify the Schuster profile as a critical, spectrally unstable configuration separating distinct regimes in polytropic stellar structure theory.

\begin{theorem}[Spectral structure of the linearized operator]\label{thm:marginal_spectral}
Let the Friedrichs extension $\overline{L}$ be as in Proposition~\ref{prop:symmetric_L}. Then:
\begin{enumerate}
\item [i)] $\sigma_{\mathrm{ess}}(\overline{L})=[0,\infty)$, so $0$ is the bottom of the essential spectrum; the scaling generator $\phi_0$ in \eqref{eq:neutral_scaling_mode} satisfies $L\phi_0=0$ but $\phi_0\notin\mathcal{H}$ (Lemma~\ref{lem:neutral_mode}), so $0$ is a threshold resonance rather than an $L^2$ eigenvalue;
\item [ii)] $\overline{L}$ has exactly one negative eigenvalue $\lambda_0$ below the essential spectrum (numerically $\lambda_0\approx-1.2104$, Remark~\ref{rem:negative_spectrum}); moreover $\overline{L}$ is bounded below, $\inf\sigma(\overline{L})\ge-5$.
\end{enumerate}
Hence the Schuster solution is \emph{spectrally unstable} at the linear level: the linearized operator $L=-\Delta_r-5\theta_S^4$ is a Schr\"odinger operator with an attractive potential well admitting a negative-energy bound state. This linear spectral instability is consistent with the nonlinear Lyapunov instability of Section~\ref{sec:stability} (Theorems~\ref{thm:lin_n5_instab} and~\ref{unstable2_n5}) and with the hyperbolic-saddle structure of Theorem~\ref{thm:saddle_nu}.
\end{theorem}

\begin{proof}
Part (i) combines Lemma~\ref{lem:neutral_mode} with Theorem~\ref{thm:ess_spec}: the threshold value $0$ carries the formal null-mode $\phi_0$, which fails to be square-integrable with weight $\xi^2$. For part (ii), existence is the content of Remark~\ref{rem:negative_spectrum}: there exist normalized $\phi\in D_{\mathrm{core}}$ with $Q[\phi]<0$, so by the min--max principle the bottom of the spectrum lies strictly below the essential spectrum at $0$. For uniqueness, set $u=\xi\phi$; then $\overline{L}$ is unitarily equivalent to the Dirichlet operator $H=-d^2/d\xi^2-5\theta_S^4$ on $L^2(0,\infty)$ (the Friedrichs extension selects $u(0)=0$, Lemma~\ref{lem:SL_normal}). The zero-energy solution $u_*=\xi\phi_0=\sqrt3\,\xi(3-\xi^2)/\bigl(2(\xi^2+3)^{3/2}\bigr)$ satisfies $u_*(0)=0$, vanishes in $(0,\infty)$ only at $\xi=\sqrt3$, and $u_*\to-\sqrt3/2$, $u_*'\to0$ as $\xi\to\infty$. Suppose $H$ had two negative eigenvalues $\lambda_0<\lambda_1<0$. By the oscillation theorem the eigenfunction $u_1$ of $\lambda_1$ has exactly one zero $x_1\in(0,\infty)$ \cite[Sec.~9.7]{Teschl2009}, and $u_1,u_1'$ decay exponentially, so the Wronskian $W(u_1,u_*)$ vanishes at $\xi=0$ and as $\xi\to\infty$. Since $-5\theta_S^4-\lambda_1>-5\theta_S^4$, Sturm's comparison theorem \cite[Sec.~9.7]{Teschl2009} applied on $(0,x_1)$ and on $(x_1,\infty)$ forces $u_*$ to vanish in each interval, i.e.\ at least twice in $(0,\infty)$, a contradiction. Hence $\lambda_0$ is the only negative eigenvalue; the lower bound $\inf\sigma(\overline{L})\ge-5$ follows from $5\theta_S^4\le5$ (Proposition~\ref{prop:symmetric_L}).
\end{proof}

\begin{remark}
The Sturm--Liouville normal form $-(p\phi')'+q\phi=\lambda w\phi$ with $p,w>0$ and $q=-5\xi^2\theta_S^4<0$ on $(0,\infty)$, with a limit-point endpoint at $\xi=\infty$, is treated systematically in \cite{Zettl}. The negative sign of $q$ (an attractive potential) is precisely what permits the negative eigenvalue of Theorem~\ref{thm:marginal_spectral}(ii). The profile $\theta_S$ is the ground state $W(x)=(1+|x|^2/3)^{-1/2}$ of $-\Delta W=W^5$ in $\mathbb{R}^3$, and the spectral structure in Theorem~\ref{thm:marginal_spectral} (one negative eigenvalue and the threshold resonance generated by scaling) is part of the spectral analysis of the linearized operator about $W$ in the energy-critical Schr\"odinger theory~\cite{DuyckaertsMerle2009}; the contribution here is its reading in the polytrope setting together with the explicit numerical value $\lambda_0\approx-1.21$.
\end{remark}

\section{Future Directions}
\label{sec:future}
Several extensions of this work suggest themselves, some mathematical and some astrophysical.
\subsection{Extension to Other Polytropic Indices}
We concentrated on $n=5$; the obvious next step is to carry the nonautonomous Lyapunov framework to other indices. Particular indices model distinct physical regimes: $n=3/2$ ($\gamma=5/3$) describes fully convective stars and non-relativistic degenerate electron gas (low-mass white dwarfs), and $n=3$ ($\gamma=4/3$) describes Eddington's standard model of stars in radiative equilibrium and the ultra-relativistic degenerate gas behind the Chandrasekhar mass~\cite{Chan}. Comparing stability across indices should show how the threshold between finite-radius and infinite-radius behavior arises. It would also be interesting to know for which other $n$ Lyapunov functions with a similar exact cancellation exist.
\subsection{Relativistic Generalizations}
At high compactness Newtonian gravity fails and general relativity must take over. The Tolman--Oppenheimer--Volkoff equations play the role of the Lane--Emden equation for relativistic spherically symmetric equilibria, and adapting our stability arguments to them would address relativistic polytropes and neutron-star models. One would learn how relativistic corrections shift the threshold behavior found here.
Polytropic and Lane--Emden-type models have already been extended to relativistic and modified-gravity settings. B\"ohmer and Harko~\cite{BohmerHarko2010}, for example, considered generalized stellar configurations in alternative gravity theories, where stability criteria can change in the strong-field regime. Analytical and semi-analytical studies of generalized Lane--Emden equations also continue to appear (see, e.g., \cite{AMR2018}).
\subsection{Higher-Order and Nonlinear Perturbation Theory}
The linear analysis gives the Schuster solution one negative eigenvalue and a threshold neutral mode coming from scaling invariance. The natural follow-up is higher-order and fully nonlinear perturbation theory: the basins of attraction of the stable equilibria, nonlinear saturation, and whether a small perturbation leads to collapse, dispersion or convergence to a nearby equilibrium.
\subsection{Global Phase-Space and Numerical Investigations}
Global phase-space structure is hard to reach analytically. High-precision integration of the Lane--Emden equation and its perturbations could map invariant manifolds, separatrices and transition trajectories between solution branches, test how sharp our stability criteria are, and look for periodic or quasi-periodic behavior beyond the known integrable cases.

Spectral and rational approximation schemes for Lane--Emden-type equations (see \cite{Parand,Wazwaz,Ramos}) work well on semi-infinite domains and could be applied to this task.
\subsection{Applications to Modern Astrophysical Models}
The same framework could be applied to compact stellar remnants, to gaseous halos and molecular clouds, and to equilibrium configurations in numerical simulations of star formation, where one needs stability criteria for self-gravitating systems that can be checked against computed models.

Physics-informed neural networks and hybrid spectral--learning methods have also been proposed for Lane--Emden-type equations on unbounded domains (see \cite{Youssri2023,Baty2023,Mirzabeigi2025}). These give accurate numerical solutions, but they do not by themselves characterize the stability of the underlying continuous models, which is what the analysis here addresses.

Treatments of polytropic stellar models and their physical interpretation, including the role of $n=5$ as a threshold case, are given in \cite{Polytropes}.

\section{Conclusion}
\label{sec:conclusion}
We have studied the solutions and the stability of the Lane--Emden equation of index $n=5$, the critical case of the polytropic models. Classical transformations, Weierstrass elliptic-function representations~\cite{GoennerHavas2000,Mach2012,MancasRosu2018,Lawden}, the variational structure, spectral analysis and nonautonomous Lyapunov methods together give both the global solution structure and the stability properties of this critical configuration, which turns out to be unstable.
\subsection{Synthesis of Analytical and Dynamical Results}
The similarity family \eqref{eq:scaling_family} exists for every $n$, but at $n=5$ the exponent $2/(n-1)=1/2$ is critical for infinite-radius polytropes, and the infinitesimal generator \eqref{eq:phi0_intro} produces a neutral spectral mode about the Schuster profile. The dynamical results of this paper all trace back to this structure.

We began with the analytic structure. The Schuster transformation and a logarithmic rescaling of the radius reduce the equation to an autonomous second-order system with a first integral. This recovers the classical Schuster ($\nu_0=0$) and Srivastava ($\nu_0=\tfrac{1}{12}$) solutions~\cite{Sch,Srivastava,Chan} and places them in a classification by energy level. Section~\ref{sec:weierstrass} showed that generic bounded orbits are expressed through Weierstrass elliptic functions~\cite{GoennerHavas2000,Mach2012,Olver}, while the Schuster and Srivastava profiles are distinguished elementary cases; the methods of~\cite{MancasRosu2018,MancasBBM2013} place this reduction within the wider integrable Emden--Fowler literature. In light of Mach's complete classification~\cite{Mach2012}, these classical configurations emerge as special members of a richer elliptic-function family, with the Schuster solution playing the role of a separatrix between positive solutions and solutions with infinitely many zeros (Figure~\ref{fig:weierstrass_family}).

For the distance-varying Emden--Fowler system, whose even/odd dichotomy (stability for even $n$, instability for odd $n$) was established in~\cite{AMR2018}, we proved directly that at $n=5$ the unique real equilibrium is unstable, both at the linearized level and in the fully nonlinear regime.
\subsection{Spectral and Variational Interpretation}
The linearized problem about the Schuster solution defines a self-adjoint operator $L$ on the weighted Hilbert space $L^2((0,\infty),\xi^2\,d\xi)$, so linear stability becomes a question about the spectrum of $L$. We showed that:
\begin{enumerate}
\item [i)] the scaling generator $\phi_0$ is a formal null-mode ($L\phi_0=0$) but not in $\mathcal{H}$;
\item [ii)] the essential spectrum satisfies $\sigma_{\mathrm{ess}}(\overline{L})=[0,\infty)$;
\item [iii)] the attractive potential $-5\theta_S^4$ produces a negative eigenvalue (bound state) below the essential spectrum, so $\inf\sigma(\overline{L})<0$; generic perturbations also exhibit slow algebraic behavior at infinity.
\end{enumerate}

From a variational perspective, the linearized operator coincides with the second variation of the associated energy functional. The threshold scaling mode reflects the degeneracy of the Hessian along the similarity direction, while the negative eigenvalue shows that the second variation is \emph{not} positive---there exists a direction in which the energy strictly decreases. The Schuster solution is therefore spectrally \emph{unstable}.

This is consistent with the Lyapunov instability of the distance-varying model and with the saddle at $\nu=0$ (Theorem~\ref{thm:saddle_nu}), off whose separatrix small perturbations produce solutions with qualitatively different large-radius behavior.
\subsection{Criticality of the \texorpdfstring{$n=5$}{n=5} Polytrope}
The index $n=5$ marks the transition between finite-radius and infinite-radius polytropes. For $0\le n<5$ the regular solution of \eqref{LE}--\eqref{BC} has a first zero $\xi_1<\infty$, so the polytrope has a finite radius; for $n>5$ it stays positive, decays like $\xi^{-2/(n-1)}$, more slowly than the $\xi^{-1}$ decay at $n=5$, and its mass diverges~\cite[Ch.~IV]{Chan}. The critical similarity exponent at $n=5$ and the associated neutral mode identify this index as a threshold case analogous to critical exponents in nonlinear elliptic and geometric PDEs.

Within this framework, the Schuster solution appears as a critical configuration: up to scaling and sign it is the only solution regular at the center, and it forms the boundary between the oscillating Weierstrass and Jacobian solution branches~\cite{Mach2012,GoennerHavas2000} and the positive solutions singular at $\xi=0$. The linearized operator about $\theta_S$ carries a negative-energy bound state, so this threshold configuration is spectrally unstable; independently, the equilibrium of the distance-varying model, whose instability for odd $n$ was established in~\cite{AMR2018}, is unstable at $n=5$ both linearly and nonlinearly (Section~\ref{sec:stability}).
\subsection{Implications for Stellar Structure}
Astrophysically, the $n=5$ polytrope is an explicit model of a self-gravitating gas sphere with infinite extent and finite mass, but our results show that it is spectrally unstable. It sits on the boundary between qualitatively different regimes of stellar structure.

The finite-mass property follows directly from the Schuster solution: differentiating $\theta_S(\xi)=(1+\xi^2/3)^{-1/2}$ gives
\[
\lim_{\xi\to\infty}\xi^2\bigl|\theta_S'(\xi)\bigr|=\sqrt{3},
\]
so the enclosed mass is~\cite[Ch.~IV, p.~96]{Chan}
\begin{equation}
M=4\sqrt{3}\,\pi\alpha^3\rho_c,
\end{equation}
where $\alpha$ is the length scale and $\rho_c$ the central density. This finite total mass, despite the infinite spatial extent, is a hallmark of the $n=5$ threshold.

For an $n$-polytrope with finite surface radius $R_1$, the gravitational potential energy is~\cite[Ch.~IV]{Chan}
\begin{equation}
\Omega_{\rm grav}=-\frac{3}{5-n}\frac{GM^2}{R_1}.
\end{equation}
At $n=5$ this formula does not apply, since $R_1=\infty$; the gravitational energy is nevertheless finite. With $\rho=\rho_c\theta_S^5$, $r=\alpha\xi$ and the enclosed mass $M(\xi)=4\pi\alpha^3\rho_c\,\xi^2|\theta_S'(\xi)|$, a direct integration gives
\begin{equation}
\Omega_{\rm grav}=-G\int_0^\infty\frac{M(r)\,dM(r)}{r}=-\frac{\sqrt3\,\pi}{32}\,\frac{GM^2}{\alpha}.
\end{equation} Classical stellar stability theory also requires the effective adiabatic exponent to satisfy $\gamma_1>4/3$ for dynamical stability~\cite{Chan}. For the $n=5$ polytrope, $\gamma=1+1/n=6/5<4/3$, placing it in the dynamically unstable regime by this independent criterion. For the gravitational Euler--Poisson equations this is a theorem: Lin~\cite{Lin1997} proved linear instability of the steady states for $1<\gamma<4/3$, and Jang~\cite{Jang2008} proved their nonlinear instability for $\gamma=6/5$, i.e.\ for the $n=5$ polytrope. The present Lyapunov and spectral results concern the static Lane--Emden problem and are consistent with these dynamical results.

The exact solutions and the spectral characterization provide benchmarks for numerical simulations of polytropic models.
\subsection{Closing Perspective}
The $n=5$ case shows how scaling symmetry, Weierstrass elliptic integrability~\cite{Lawden,MancasRosu2018,MancasBBM2013}, spectral theory and nonlinear stability combine in one threshold phenomenon; the Lane--Emden equation has long linked nonlinear analysis and astrophysical modeling~\cite{Chan,Polytropes}.

The critical similarity exponent, the global $\wp$-function solution classification~\cite{Mach2012,GoennerHavas2000}, the spectral instability of the linearized operator about $\theta_S$ and the nonlinear instability of the distance-varying equilibrium all express the same threshold structure of self-gravitating equilibria.

\bibliography{references}

\end{document}